\documentclass[pra,twocolumn]{revtex4-1}%
\usepackage{amsfonts}
\usepackage{amsmath}
\usepackage{amssymb}
\usepackage{graphicx}
\usepackage{hyperref}
\usepackage[caption=false]{subfig}%
\providecommand{\U}[1]{\protect\rule{.1in}{.1in}}
\newtheorem{theorem}{Theorem}

\newtheorem{proposition}[theorem]{Proposition}

\newenvironment{proof}[1][Proof]{\noindent\textbf{#1.} }{\ \rule{0.5em}{0.5em}}
\begin{document}
\title{Upper bounds on secret key agreement over lossy thermal bosonic channels}
\author{Eneet Kaur}
\affiliation{Hearne Institute for Theoretical Physics, Department of Physics and Astronomy,
Baton Rouge, Louisiana 70803, USA}
\author{Mark M. Wilde}
\affiliation{Hearne Institute for Theoretical Physics, Department of Physics and Astronomy,
Baton Rouge, Louisiana 70803, USA}
\affiliation{Center for Computation and Technology, Louisiana State University, Baton
Rouge, Louisiana 70803, USA}

\begin{abstract}
Upper bounds on the secret-key-agreement capacity of a quantum channel serve
as a way to assess the performance of practical quantum-key-distribution
protocols conducted over that channel. In particular, if a protocol employs a
quantum repeater, achieving secret-key rates exceeding these upper bounds is a
witness to having a working quantum repeater. In this paper, we extend a
recent advance [Liuzzo-Scorpo \textit{et al}., Phys.~Rev.~Lett.~119, 120503 (2017)] in the theory
of the teleportation simulation of single-mode phase-insensitive Gaussian
channels such that it now applies to the relative entropy of entanglement
measure. As a consequence of this extension, we find tighter upper bounds on
the non-asymptotic secret-key-agreement capacity of the lossy thermal bosonic
channel than were previously known. The lossy thermal bosonic channel serves
as a more realistic model of communication than the pure-loss bosonic channel,
because it can model the effects of eavesdropper tampering and imperfect
detectors. An implication of our result is that the previously known upper
bounds on the secret-key-agreement capacity of the thermal channel are too
pessimistic for the practical finite-size regime in which the channel is used
a finite number of times, and so it should now be somewhat easier to witness a
working quantum repeater when using secret-key-agreement capacity upper bounds
as a benchmark.

\end{abstract}
\date{\today}
\maketitle
\preprint{ }
\volumeyear{year}
\volumenumber{number}
\issuenumber{number}
\eid{identifier}
\startpage{1}
\endpage{102}

\section{Introduction}

One of the main goals of quantum information theory \cite{H13book,H06,W15book}%
\ is to establish bounds on communication rates for various
information-processing tasks. An important application lies in the domain of
secret communication, following the development of quantum key distribution
\cite{bb84,E91}. In recent years, there has been a growing interest in
establishing bounds on the secret-key-agreement capacity of a quantum channel,
which is the highest rate at which communicating parties can use the channel
and public classical communication to distill a secret key
\cite{TGW14IEEE,TGW14Nat,STW16,PLOB15,Goodenough2015,TSW16,AML16,WTB16,Christandl2017,Wilde2016,BA17,RGRKVRHWE17,TSW17}%
. Such bounds have been proven by exploiting the methods of quantum
information theory and can be interpreted as setting the fundamental
limitations of quantum key distribution whenever a quantum repeater is not
available \cite{L15}.

An important development occurred in \cite{TGW14Nat}, in which it was
established that there is a fundamental rate-loss trade-off that any
repeaterless quantum key distribution protocol cannot overcome. That is,
without a quantum repeater, the rate of secret key that can be distilled from
a pure-loss bosonic channel (lossy optical fiber or a free-space channel)
decreases exponentially with increasing distance \cite{TGW14Nat}.

Later, this bound was improved to establish that the secret-key-agreement
capacity of a pure-loss bosonic channel of transmissivity $\eta\in(0,1)$ is
equal to $-\log_{2}(1-\eta)$. This bound was claimed in \cite{PLOB15} and
rigorously proven in \cite{WTB16}. In particular, let $P_{\mathcal{L}_{\eta}%
}^{\leftrightarrow}(n,\varepsilon)$ denote the highest rate at which
$\varepsilon$-close-to-ideal secret key can be distilled by making $n$
invocations of a pure-loss channel $\mathcal{L}_{\eta}$ of transmissivity
$\eta$, along with the assistance of public classical communication
\cite{WTB16}. In \cite{WTB16}, $P_{\mathcal{L}_{\eta}}^{\leftrightarrow
}(n,\varepsilon)$ is called the non-asymptotic secret-key-agreement capacity
of the channel $\mathcal{L}_{\eta}$. One of the results of \cite{WTB16} is the
following fundamental upper bound:%
\begin{equation}
P_{\mathcal{L}_{\eta}}^{\leftrightarrow}(n,\varepsilon)\leq-\log_{2}%
(1-\eta)+\frac{C(\varepsilon)}{n}, \label{eq:SC-bound}%
\end{equation}
where $C(\varepsilon)=\log_{2}6+2\log_{2}(\left[  1+\varepsilon\right]
/\left[  1-\varepsilon\right]  )$. The bound in \eqref{eq:SC-bound} is known
as a strong converse bound because it converges to the secret-key-agreement
capacity $-\log_{2}(1-\eta)$ in the limit as $n\rightarrow\infty$. We suspect
that there is little room for improvement of the bound in \eqref{eq:SC-bound}
and discuss this point further in Appendix~\ref{sec:disp-rev-coh-info}. The
bound in \eqref{eq:SC-bound} is to be contrasted with the following
weak-converse bound:%
\begin{equation}
P_{\mathcal{L}_{\eta}}^{\leftrightarrow}(n,\varepsilon)\leq\frac
{1}{1-\varepsilon}\left[  -\log_{2}(1-\eta)+\frac{h_{2}(\varepsilon)}%
{n}\right]  , \label{eq:WC-bound}%
\end{equation}
which follows as a direct consequence of \cite[Section~8]{WTB16} and
\cite[Eq.~(2)]{WR12} (see also \cite[Eq.~(134)]{MW12}). For the benefit of the
reader, we explain how to arrive at this weak-converse bound in more detail in
Appendix~\ref{sec:weak-converse-bnds}. In the above,%
\begin{equation}
h_{2}(\varepsilon)=-\varepsilon\log_{2}\varepsilon-(1-\varepsilon)\log
_{2}(1-\varepsilon)
\end{equation}
denotes the binary entropy. The bound in \eqref{eq:WC-bound} is a
weak-converse bound because it requires the extra limit as $\varepsilon
\rightarrow0$ after taking the limit as $n\rightarrow\infty$, in order to
arrive at the capacity upper bound of $-\log_{2}(1-\eta)$. The significance of
the bounds in \eqref{eq:SC-bound} and \eqref{eq:WC-bound} is that they apply
for any finite number $n$ of channel uses and key-quality
parameter$~\varepsilon$. As such, these bounds can be used to assess the
performance of any practical secret-key-agreement protocol conducted over a
pure-loss channel$~\mathcal{L}_{\eta}$.

The pure-loss channel is somewhat of an ideal model for a communication
channel, even if it does have a strong physical underpinning in the context of
free-space communication \cite{YS78,S09}. In particular, a working assumption
of the model is that the channel input interacts with an environment prepared
in the vacuum state. However, in practical setups, we might expect the
environment to be modeled as a thermal state of a fixed mean photon number
$N_{B}>0$ \cite{S09}, and in such a case, the channel is called a thermal
channel and denoted by $\mathcal{L}_{\eta,N_{B}}$\ (also called thermal-lossy
channel, as in \cite{RGRKVRHWE17}). This added thermal noise is often called
excess noise \cite{NH04,LDTG05}, which can serve as a simple model of
tampering by an eavesdropper. Additionally, there are realistic effects in
communication schemes, such as dark counts of photon detectors that can be
modeled as arising from thermal photons in the environment
\cite{S09,RGRKVRHWE17}. As such, it is an important goal to establish upper
bounds on the secret-key-agreement capacity of the thermal channel in order to
assess the performance of practical secret-key-agreement protocols, and the
main contribution of the present paper is to establish upper bounds on the
non-asymptotic secret-key-agreement capacity $P_{\mathcal{L}_{\eta,N_{B}}%
}^{\leftrightarrow}(n,\varepsilon)$ of the thermal channel $\mathcal{L}%
_{\eta,N_{B}}$, which improve upon the prior known bounds from
\cite{PLOB15,WTB16}\ in certain regimes.

Prior works established that%
\begin{equation}
-\log_{2}(\left[  1-\eta\right]  \eta^{N_{B}})-g(N_{B})
\label{eq:rel-ent-thermal}%
\end{equation}
is an upper bound on the secret-key-agreement capacity of a thermal channel
$\mathcal{L}_{\eta,N_{B}}$ with transmissivity $\eta\in(0,1)$ and thermal mean
photon number $N_{B}>0$. This bound was claimed in \cite{PLOB15} and
rigorously proven in \cite{WTB16}. In this expression,%
\begin{equation}
g(N_{B})=(N_{B}+1)\log_{2}(N_{B}+1)-N_{B}\log_{2}N_{B}%
\end{equation}
is the entropy of a thermal state of mean photon number~$N_{B}$. In
particular, the following bound was given in \cite[Section~8]{WTB16}%
\begin{multline}
P_{\mathcal{L}_{\eta,N_{B}}}^{\leftrightarrow}(n,\varepsilon)\leq-\log
_{2}(\left[  1-\eta\right]  \eta^{N_{B}})-g(N_{B})\label{eq:SC-thermal}\\
+\sqrt{\frac{2V_{\eta,N_{B}}}{n\left(  1-\varepsilon\right)  }}+\frac
{C(\varepsilon)}{n},
\end{multline}
where%
\begin{equation}
V_{\eta,N_{B}}=N_{B}(N_{B}+1)\log_{2}^{2}(\eta\left[  N_{B}+1\right]  /N_{B}),
\end{equation}
and the following weak-converse bound is a direct consequence of
\cite[Section~8]{WTB16} and \cite{WR12,MW12} (explained also in
Appendix~\ref{sec:weak-converse-bnds}):%
\begin{multline}
P_{\mathcal{L}_{\eta,N_{B}}}^{\leftrightarrow}(n,\varepsilon)\leq
\label{eq:WC-thermal}\\
\frac{1}{1-\varepsilon}\left[  -\log_{2}(\left[  1-\eta\right]  \eta^{N_{B}%
})-g(N_{B})+\frac{h_{2}(\varepsilon)}{n}\right]  .
\end{multline}
Again, the value of these bounds is that they apply for any finite number $n$
of channel uses and key-quality parameter $\varepsilon$. However, by
inspecting \eqref{eq:SC-thermal}, we see that the order $1/\sqrt{n}$ and lower
terms are strictly positive.

The main contribution of the present paper is to improve the bound in
\eqref{eq:SC-thermal} in such a way that the order $1/\sqrt{n}$ term is
negative whenever $\varepsilon<1/2$, representing the backoff from capacity
incurred by using the channel a finite number of times while allowing for
non-zero error. In fact, we find the following improved bound for several
realistic values of $\eta$ and $N_{B}$:%
\begin{multline}
P_{\mathcal{L}_{\eta,N_{B}}}^{\leftrightarrow}(n,\varepsilon)\leq-\log
_{2}(\left[  1-\eta\right]  \eta^{N_{B}})-g(N_{B})\label{eq:new-bound}\\
+\sqrt{\frac{V_{\eta,N_{B}}^{\prime}}{n}}\Phi^{-1}(\varepsilon)+\frac{O(\log
n)}{n},
\end{multline}
where $V_{\eta,N_{B}}^{\prime}$ is a channel-dependent parameter that we
discuss later and $\Phi^{-1}$ denotes the inverse of the cumulative normal
distribution function (see \eqref{eq:Phi-inv}), for which we have that
$\Phi^{-1}(\varepsilon)<0$ whenever $\varepsilon<1/2$. We should note that the
bound in \eqref{eq:new-bound}\ applies only for $n$ sufficiently large (such
that $n$ is proportional to $1/\varepsilon^{2}$), as it relies on the
Berry--Esseen theorem \cite{KS10,S11}, but many prior works have shown that
first- and second-order terms like the above one serve as an excellent
approximation for non-asymptotic capacities even for small $n$
\cite{polyanskiy10,P10,MW12,P13,TBR15}. The main new tool that we use to
establish this result, beyond those used and introduced in \cite{WTB16}, is a
recent development in \cite{LMGA17} regarding teleportation simulation of
single-mode phase-insensitive bosonic channels using finite-energy resource
states. Figure~\ref{fig:results} plots this bound for several realistic values
of the distance $L$ (related to transmissivity $\eta$) and thermal mean photon
number $N_{B}$, and we point to Section~\ref{sec:results}\ for a more
detailed discussion of these figures.

In the remainder of the paper, we argue how to arrive at the bound in
\eqref{eq:new-bound}. In what follows, we review the formalism of quantum
Gaussian states and channels \cite{adesso14,S17}, and we also review
information quantities needed, such as quantum relative entropy and relative
entropy variance. We then review the critical tool of teleportation simulation
of a quantum channel \cite{BDSW96,WPG07,NFC09,Mul12} and how it can be used
with \cite[Eq.~(4.34)]{WTB16} and ideas from \cite{LMGA17} in order to arrive
at \eqref{eq:new-bound}. We finally close with a summary and some open questions.

\section{Preliminaries}

\subsection{Quantum Gaussian states and channels}

The main class of quantum states in which we are interested in this paper are
quantum Gaussian states \cite{adesso14,S17}. In our brief review, we consider
$m$-mode Gaussian states, where $m$ is some fixed positive integer. Let
$\hat{x}_{j}$ denote each quadrature operator ($2m$ of them for an $m$-mode
state), and let $\hat{x}\equiv\left[  \hat{q}_{1},\ldots,\hat{q}_{m},\hat
{p}_{1},\ldots,\hat{p}_{m}\right]  \equiv\left[  \hat{x}_{1},\ldots,\hat
{x}_{2m}\right]  $ denote the vector of quadrature operators, so that the
first~$m$ entries correspond to position-quadrature operators and the last~$m$
to momentum-quadrature operators. The quadrature operators satisfy the
following commutation relations:%
\begin{equation}
\left[  \hat{x}_{j},\hat{x}_{k}\right]  =i\Omega_{j,k},\quad\mathrm{where}%
\quad\Omega=%
\begin{bmatrix}
0 & 1\\
-1 & 0
\end{bmatrix}
\otimes I_{m}\text{,} \label{eq:symplectic-form}%
\end{equation}
and $I_{m}$ is the $m\times m$ identity matrix. We also take the annihilation
operator $\hat{a}=\left(  \hat{q}+i\hat{p}\right)  /\sqrt{2}$. Let $\rho$ be a
Gaussian state, with the mean-vector entries $\left\langle \hat{x}%
_{j}\right\rangle ^{\rho}=\mu_{j}^{\rho}$, and let $\mu^{\rho}$ denote the
mean vector. The entries of the covariance matrix $V^{\rho}$\ of $\rho$ are
given by
\begin{equation}
V_{j,k}^{\rho}\equiv\left\langle \left\{  \hat{x}_{j}-\mu_{j}^{\rho},\hat
{x}_{k}-\mu_{k}^{\rho}\right\}  \right\rangle ^{\rho}.
\label{eq:covariance-matrices}%
\end{equation}
A $2m\times2m$ matrix $S$ is symplectic if it preserves the symplectic form:
$S\Omega S^{T}=\Omega$. According to Williamson's theorem \cite{W36}, there is
a diagonalization of the covariance matrix $V^{\rho}$ of the form,
\begin{equation}
V^{\rho}=S^{\rho}\left(  D^{\rho}\oplus D^{\rho}\right)  \left(  S^{\rho
}\right)  ^{T},
\end{equation}
where $S^{\rho}$ is a symplectic matrix and $D^{\rho}\equiv\operatorname{diag}%
(\nu_{1},\ldots,\nu_{m})$ is a diagonal matrix of symplectic eigenvalues such
that $\nu_{i}\geq1$ for all $i\in\left\{  1,\ldots,m\right\}  $. We say that a
quantum Gaussian state is faithful if all of its symplectic eigenvalues are
strictly greater than one (this also means that the state is positive
definite). Faithfulness of Gaussian states is required to ensure that
$G^{\rho}$ is non-singular. We can write the density operator $\rho$ of a
faithful state in the following exponential form \cite{PhysRevA.71.062320,K06,H10} (see also \cite{H13book,S17})%
:
\begin{align}
&  \rho=(Z^{\rho})^{-1/2}\exp\left[  -\frac{1}{2}(\hat{x}-\mu^{\rho}%
)^{T}G^{\rho}(\hat{x}-\mu^{\rho})\right]  ,\label{eq:exp-form}\\
&  \mathrm{with}\quad Z^{\rho}\equiv\det(\left[  V^{\rho}+i\Omega\right]
/2)\\
&  \mathrm{and}\quad G^{\rho}\equiv-2\Omega S^{\rho}\left[
\operatorname{arcoth}(D^{\rho})\right]  ^{\oplus2}\left(  S^{\rho}\right)
^{T}\Omega, \label{eq:G_rho}%
\end{align}
where $\operatorname{arcoth}(x)\equiv\frac{1}{2}\ln\!\left(  \frac{x+1}%
{x-1}\right)  $ with domain $\left(  -\infty,-1\right)  \cup\left(
1,+\infty\right)  $.
Note that we can also write%
\begin{equation}
G^{\rho}=2i\Omega\operatorname{arcoth}(iV^{\rho}\Omega),
\label{eq:more-compact-G-rho}%
\end{equation}
so that $G^{\rho}$ is represented directly in terms of the covariance matrix
$V^{\rho}$. By inspection, the $G$ and $V$ matrices are symmetric. In what
follows, we adopt the same notation for quantities associated with a density
operator $\sigma$, such as $\mu^{\sigma}$, $V^{\sigma}$, $S^{\sigma}$,
$D^{\sigma}$, $Z^{\sigma}$, and $G^{\sigma}$.

A two-mode Gaussian state $\rho$ with covariance matrix in \textquotedblleft
standard form\textquotedblright\ has a covariance matrix as follows
\cite{DGCZ00,S00}:%
\begin{equation}
V^{\rho}=%
\begin{bmatrix}
a & c\\
c & b
\end{bmatrix}
\oplus%
\begin{bmatrix}
a & -c\\
-c & b
\end{bmatrix}
. \label{eq:CM-standard-form}%
\end{equation}
The symplectic diagonalization of the covariance matrix~$V$ simplifies as well
\cite{SIS04}:%
\begin{equation}
V=S\left(  D\oplus D\right)  S^{T},
\end{equation}
where%
\begin{align}
S  &  =\left(  I_{2}\oplus\sigma_{Z}\right)  S_{0}^{\oplus2}\left(
I_{2}\oplus\sigma_{Z}\right)  ,\\
S_{0}  &  =%
\begin{bmatrix}
\omega_{+} & \omega_{-}\\
\omega_{-} & \omega_{+}%
\end{bmatrix}
,\qquad\omega_{\pm} =\sqrt{\frac{a+b\pm\sqrt{y}}{2\sqrt{y}}},\\
D  &  =%
\begin{bmatrix}
\nu_{-} & 0\\
0 & \nu_{+}%
\end{bmatrix}
,\qquad\nu_{\pm} =\left[  \sqrt{y}\pm\left(  b-a\right)  \right]  /2,\\
y  &  =\left(  a+b\right)  ^{2}-4c^{2},
\end{align}
and $\sigma_{Z}$ denotes the standard Pauli $Z$ matrix. Given a two-mode state
with covariance matrix in standard form as in \eqref{eq:CM-standard-form}, it
is a separable state if%
\begin{equation}
c\leq c_{\text{sep}}\equiv\sqrt{\left(  a-1\right)  \left(  b-1\right)  },
\label{eq:c_sep}%
\end{equation}
which can be determined from the condition given in \cite[Eq.~(14)]{AI05}. We
return to this condition when we discuss the relative entropy of entanglement
for quantum Gaussian states.

A quantum Gaussian channel is one that preserves Gaussian states
\cite{CEGH08,adesso14,S17}. The action of a quantum Gaussian channel on an
input state $\rho$ is characterized by two matrices $X$ and $Y$, which
transform the covariance matrix $V^{\rho}$ of $\rho$ as follows:%
\begin{equation}
V^{\rho}\rightarrow XV^{\rho}X^{T}+Y, \label{eq:Gaussian-channel}%
\end{equation}
where $X^{T}$ is the transpose of the matrix $X$. In this formalism, the
thermal channel $\mathcal{L}_{\eta,N_{B}}$ with transmissivity $\eta\in(0,1)$
and thermal mean photon number $N_{B}>0$ is given by%
\begin{equation}
X=\sqrt{\eta}I_{2},\qquad Y=(1-\eta)(2N_{B}+1)I_{2},
\label{eq:thermal-ch-Gaussian}%
\end{equation}
where $I_{2}$ is the $2\times2$ identity matrix. Our principal focus in this
paper is on the thermal channel.

\subsection{Teleportation simulation and reduction by teleportation}

\label{sec:tp-simulation}Teleportation simulation of a channel
\cite{BDSW96,WPG07,NFC09,Mul12}\ is a key tool used to establish the upper
bounds in \eqref{eq:SC-bound}, \eqref{eq:WC-bound}, \eqref{eq:SC-thermal}, and
\eqref{eq:WC-thermal}. The basic idea behind this tool is that channels with
sufficient symmetry can be simulated by the action of a teleportation protocol
\cite{PhysRevLett.70.1895,prl1998braunstein,Werner01}\ on a resource state
$\omega_{AB}$\ shared between the sender $A$\ and receiver $B$. More
generally, a channel $\mathcal{N}_{A^{\prime}\rightarrow B}$ with input system
$A^{\prime}$ and output system $B$ is defined to be teleportation simulable
with associated resource state $\omega_{AB}$ if the following equality holds
for all input states $\rho_{A^{\prime}}$:%
\begin{equation}
\mathcal{N}_{A^{\prime}\rightarrow B}(\rho_{A^{\prime}})=\mathcal{T}%
_{A^{\prime}AB}(\rho_{A^{\prime}}\otimes\omega_{AB}%
),\label{eq:TP-simulable-channel}%
\end{equation}
where $\mathcal{T}_{A^{\prime}AB}$ is a quantum channel consisting of local
operations and classical communication between the sender, who has systems
$A^{\prime}$ and $A$, and the receiver, who has system $B$ ($\mathcal{T}%
_{A^{\prime}AB}$ can also be considered a generalized teleportation protocol,
as in \cite{Werner01}). The definition in \eqref{eq:TP-simulable-channel} was
first given in \cite{WTB16arxiv}, based on many earlier developments
\cite{BDSW96,Werner01,WPG07,NFC09,Mul12}. The implication of channel
simulation via teleportation is that the performance of a general protocol
that uses the channel $n$ times, with each use interleaved by local operations
and classical communication (LOCC), can be bounded from above by the
performance of a protocol with a much simpler form:\ the simplified protocol
consists of a single round of LOCC acting on $n$ copies of $\omega_{AB}$
\cite{BDSW96,NFC09,Mul12}. This is called reduction by teleportation. Of
course, a secret-key-agreement protocol is one particular kind of protocol of
the above form, as considered in \cite{PLOB15,WTB16}, and so the general
reduction method of \cite{BDSW96,NFC09,Mul12}\ applies.

For continuous-variable bosonic systems, the teleportation simulation of a
single-mode bosonic Gaussian channels was considered in \cite{NFC09}, and the
simulation therein only simulates the channel exactly in the limit in which
the resource state is the result of transmitting one share of an
infinitely-squeezed, two-mode squeezed vacuum state \cite{S17} through the
channel (this resource state is sometimes called the Choi state of the channel
\cite{S17}, and we use this terminology in what follows). Thus, when applying
this argument to bound the rates of secret-key-agreement protocols as
discussed above, one must take care with an appropriate limiting argument, as
pointed out in \cite{N16}\ and handled already in \cite{WTB16}. This
teleportation simulation argument with an infinitely-squeezed resource state
is one of the core steps used to establish the bounds in \eqref{eq:SC-bound},
\eqref{eq:WC-bound}, \eqref{eq:SC-thermal}, and \eqref{eq:WC-thermal}.

Recently, an important development in the theory of the teleportation
simulation of quantum Gaussian channels has taken place \cite{LMGA17}. In
particular, the authors of \cite{LMGA17}\ have shown that all single-mode,
phase-insensitive quantum Gaussian channels other than the pure-loss channel
can be simulated via the action of teleportation on a finite-energy Gaussian
resource state that has the same amount of entanglement as the Choi state of
the channel. In \cite{LMGA17}, the authors quantified the amount of
entanglement in the resource state using an entanglement monotone
\cite{H42007} called logarithmic negativity, which is the same entanglement
measure considered in \cite{NFC09}. In our paper, we show how the main idea of
their paper leads to strengthened bounds on the performance of
secret-key-agreement protocols conducted over single-mode phase-insensitive
bosonic Gaussian channels.

To describe the result of \cite{LMGA17} in more detail, let $X=\sqrt{\tau
}I_{2}$ and $Y=yI_{2}$ be the matrices representing the action of a
single-mode phase-insensitive Gaussian channel on an input state, as in
\eqref{eq:Gaussian-channel}. In what follows and as in \cite{LMGA17}, we
exclusively consider the case when $\tau\geq0$. In order for the map to be a
completely positive, trace-preserving map (i.e., a legitimate quantum
channel), the following inequality should hold \cite{S17}%
\begin{equation}
y\geq\left\vert 1-\tau\right\vert .
\end{equation}
The main contribution of \cite{LMGA17} is that every single-mode
phase-insensitive Gaussian channel in the above class, besides the pure-loss
channel, can be simulated by the action of a continuous-variable teleportation
protocol on a finite-energy, two-mode resource state with the same amount of
entanglement as the Choi state of the channel. An additional contribution of
\cite{LMGA17} is a converse bound:\ it is not possible to use a resource state
with logarithmic negativity smaller than that of the Choi state, in order to
simulate the channel. This follows directly from the facts that the
teleportation simulation protocol should simulate the channel, teleportation
is an LOCC, and logarithmic negativity is an entanglement monotone (it is
non-increasing with respect to an LOCC). This converse bound holds, by the
same argument, for all measures of entanglement (such as relative entropy of entanglement).

In more detail, the teleportation simulation of \cite{LMGA17} begins with the
sender and receiver of the channel sharing a two-mode Gaussian state in the
standard form in \eqref{eq:CM-standard-form}. The sender mixes the input of
the channel and her share of the resource state on a 50-50 beam splitter. The
sender then performs ideal homodyne detection of the position quadrature of
the first mode and ideal homodyne detection of the momentum quadrature of the
second mode, leading to measurement outcomes $Q_{+}$ and $P_{-}$. The sender
communicates these real values over ideal classical communication channels to
the receiver, and the receiver performs displacement operations of his mode by
$g\sqrt{2}Q_{+}$ and $g\sqrt{2}P_{-}$, for some $g\in\mathbb{R}$. The result
of all of these operations is to implement a quantum Gaussian channel of the
following form on the input state:%
\begin{align}
X  &  =gI_{2},\\
Y  &  =\left[  g^{2}a+2gc+b\right]  I_{2},
\end{align}
where we note the different sign convention from \cite[Eq.~(7)]{LMGA17}, due
to our slightly different convention for the standard form in
\eqref{eq:CM-standard-form}. If $g>0$, then the channel implemented is a
single-mode phase-insensitive Gaussian channel with%
\begin{equation}
\tau=g^{2},\qquad y=g^{2}a+2gc+b. \label{eq:TP-simulation-1-mode-G}%
\end{equation}
If $g<0$, then one can postprocess the output according to a unitary Gaussian
channel with $X=-I_{2}$ and $Y=0$ (a phase flip channel), such that the
overall channel is a single-mode phase-insensitive Gaussian channel with
$\tau$ and $y$ as in \eqref{eq:TP-simulation-1-mode-G}. A generalization of
these steps beyond two-mode states is given in \cite{WPG07}.

Where \cite{LMGA17} departs from prior works is to solve an inverse problem
regarding teleportation simulation. Given values of $\tau$ and $y$
corresponding to a physical channel different from the pure-loss channel, the
authors of \cite{LMGA17} proved that there exists a finite-energy, two-mode
Gaussian state in standard form satisfying \eqref{eq:TP-simulation-1-mode-G},
having its smaller symplectic eigenvalue equal to one, and having its
logarithmic negativity equal to that of the Choi state of the channel. It
should be stressed that the states found in \cite{LMGA17} have an analytical
form, which has to do with the form of the above constraints.

\subsection{Information quantities and bounds for secret-key-agreement
protocols}

The basic information quantities that we need in this paper are the quantum
relative entropy \cite{U62,Lindblad1973}, the relative entropy variance
\cite{li12,TH12}, and the hypothesis testing relative entropy \cite{BD10,WR12}%
. For two states $\rho$ and $\sigma$ defined on a separable Hilbert space with
the following spectral decompositions:%
\begin{align}
\rho &  =\sum_{x}\lambda_{x}|\phi_{x}\rangle\langle\phi_{x}%
|,\label{eq:spec-decomp-rho}\\
\sigma &  =\sum_{y}\mu_{y}|\psi_{y}\rangle\langle\psi_{y}|,
\label{eq:spec-decomp-sigma}%
\end{align}
the quantum relative entropy $D(\rho\Vert\sigma)$ \cite{Lindblad1973}\ and the
relative entropy variance $V(\rho\Vert\sigma)$\ \cite{li12,TH12} are defined
as
\begin{align}
D(\rho\Vert\sigma)  &  =\sum_{x,y}\left\vert \left\langle \psi_{y}|\phi
_{x}\right\rangle \right\vert ^{2}\lambda_{x}\log_{2}(\lambda_{x}/\mu_{y}),\\
V(\rho\Vert\sigma)  &  =\sum_{x,y}\left\vert \left\langle \psi_{y}|\phi
_{x}\right\rangle \right\vert ^{2}\lambda_{x}\left[  \log_{2}(\lambda_{x}%
/\mu_{y})-D(\rho\Vert\sigma)\right]  ^{2}.
\end{align}

For quantum Gaussian states, the quantities $D(\rho\Vert\sigma)$ \cite{PhysRevA.71.062320}, \cite{PLOB15} and
$V(\rho\Vert\sigma)$ \cite{BLTW16} can be expressed in terms of their first and second
moments. For simplicity, let us
suppose that $\rho$ and $\sigma$ are zero-mean quantum Gaussian states. Then
Refs.~\cite{PhysRevA.71.062320}, \cite{PLOB15} established that%
\begin{equation}
D(\rho\Vert\sigma)=\log_{2}(Z^{\sigma}/Z^{\rho})/2-\operatorname{Tr}\{\Delta
V^{\rho}\}/4\ln2, \label{eq:rel-ent-G}%
\end{equation}
where $\Delta=G^{\rho}-G^{\sigma}$, and Ref.~\cite{BLTW16} established that%
\begin{equation}
V(\rho\Vert\sigma)=\frac{1}{8\ln^{2}2}\left[  \operatorname{Tr}\{\Delta
V^{\rho}\Delta V^{\rho}\}+\operatorname{Tr}\{\Delta\Omega\Delta\Omega
\}\right]  . \label{eq:rel-ent-var-G}%
\end{equation}
In the above, we should note that our convention for normalization of
covariance matrices is what leads to the different constant prefactors when
compared to the expressions in \cite{PhysRevA.71.062320,PLOB15,BLTW16}.

The hypothesis testing relative entropy is defined as \cite{BD10,WR12}%
\begin{multline}
D_{H}^{\varepsilon}(\rho\Vert\sigma)=\\
-\log_{2}\inf_{\Lambda}\{\operatorname{Tr}\{\Lambda\sigma\}:0\leq\Lambda\leq
I\wedge\operatorname{Tr}\{\Lambda\rho\}\geq1-\varepsilon\}
\end{multline}
By the reasoning in \cite{DPR15} and Appendix~\ref{sec:AEP-hypo}, we have the
following bound holding for faithful states $\rho$ and $\sigma$ such that
$D(\rho\Vert\sigma),V(\rho\Vert\sigma),T(\rho\Vert\sigma)<\infty$ and
$V(\rho\Vert\sigma)>0$:%
\begin{multline}
D_{H}^{\varepsilon}(\rho^{\otimes n}\Vert\sigma^{\otimes n})\leq\\
nD(\rho\Vert\sigma)+\sqrt{nV(\rho\Vert\sigma)}\Phi^{-1}(\varepsilon)+O(\log
n), \label{eq:upper-bnd-hypo}%
\end{multline}
where \cite{li12,TH12}%
\begin{equation}
T(\rho\Vert\sigma)=\sum_{x,y}\left\vert \left\langle \psi_{y}|\phi
_{x}\right\rangle \right\vert ^{2}\lambda_{x}\left\vert \log_{2}\left(
\lambda_{x}/\mu_{y}\right)  -D(\rho\Vert\sigma)\right\vert ^{3},
\end{equation}
and
\begin{align}
\Phi(a)  &  =\frac{1}{\sqrt{2\pi}}\int_{-\infty}^{a} dx \ \exp\!\left(
\frac{-x^{2}}{2}\right) \\
\Phi^{-1}(\varepsilon)  &  =\sup\left\{  a\in\mathbb{R} \ |\ \Phi(a)
\leq\varepsilon\right\}  . \label{eq:Phi-inv}%
\end{align}
We note here that the finiteness of $T(\rho\Vert\sigma)$ for finite-energy,
faithful Gaussian states is essential to the main result of our paper.
Inspecting the proof given in Appendix~\ref{sec:AEP-hypo}, we see that the
condition $T(\rho\Vert\sigma) < \infty$ allows us to invoke the Berry-Esseen
theorem \cite{KS10,S11}, which in turn leads to the improved upper bound in
\eqref{eq:new-bound}.

The relative entropy of entanglement of a bipartite state $\rho_{AB}$ is
defined as follows \cite{VP98}:%
\begin{equation}
E_{R}(A;B)_{\rho}=\inf_{\sigma_{AB}\in\operatorname{SEP}(A:B)}D(\rho_{AB}%
\Vert\sigma_{AB}),
\end{equation}
where $\operatorname{SEP}(A\!:\!B)$ denotes the set of separable (unentangled)
states \cite{W89}. Analogously, we have the $\varepsilon$-relative entropy of
entanglement \cite{BD11}:%
\begin{equation}
E_{R}^{\varepsilon}(A;B)_{\rho}=\inf_{\sigma_{AB}\in\operatorname{SEP}%
(A:B)}D_{H}^{\varepsilon}(\rho_{AB}\Vert\sigma_{AB}).
\label{eq:eps-rel-entr-enta}%
\end{equation}
For a two-mode Gaussian state $\rho_{AB}$ in standard form, one can always
choose the separable state $\sigma_{AB}^{\prime}$ to be in standard form with
the same values for $a$ and $b$ but with $c$ chosen to saturate the inequality
in \eqref{eq:c_sep}, such that $c=c_{\text{sep}}$ \cite{PLOB15}. By
definition, for this suboptimal choice, we have that%
\begin{align}
E_{R}(A;B)_{\rho}  &  \leq D(\rho_{AB}\Vert\sigma_{AB}^{\prime}),\\
E_{R}^{\varepsilon}(A;B)_{\rho}  &  \leq D_{H}^{\varepsilon}(\rho_{AB}%
\Vert\sigma_{AB}^{\prime}),
\end{align}
and this is the choice made in \cite{PLOB15,WTB16}\ to arrive at various upper
bounds on secret-key-agreement capacity. In what follows, we refer to
$D(\rho_{AB}\Vert\sigma_{AB}^{\prime})$ as the suboptimal relative entropy of
entanglement of $\rho_{AB}$.

In \cite[Eq.~(4.34)]{WTB16}, the following bound was established on the
non-asymptotic secret-key-agreement capacity of a channel $\mathcal{N}$\ that
is teleportation simulable with associated resource state $\omega_{AB}$:%
\begin{equation}
\label{eqn:hyptest}P_{\mathcal{N}}^{\leftrightarrow}(n,\varepsilon)\leq
\frac{1}{n}E_{R}^{\varepsilon}(A^{n};B^{n})_{\omega^{\otimes n}}\leq\frac
{1}{n}D_{H}^{\varepsilon}(\omega_{AB}^{\otimes n}\Vert\sigma_{AB}^{\otimes
n}).
\end{equation}
The argument for the first inequality critically relies upon the connection
between secret-key-agreement protocols and private-state distillation
protocols established in \cite{HHHO05,HHHO09}\ and some other results
contained therein, in addition to the teleportation reduction argument
discussed in Section~\ref{sec:tp-simulation}. The second inequality follows
from the definition in \eqref{eq:eps-rel-entr-enta}, with $\sigma_{AB}$ being
an arbitrary separable state. Thus, any resource state for the teleportation
simulation of a channel can be used to give an upper bound on its
non-asymptotic secret-key-agreement capacity. In particular, if $\omega_{AB}$
and $\sigma_{AB}$ are faithful quantum Gaussian states of finite energy such
that $\omega_{AB}\neq\sigma_{AB}$, then the conditions $D(\omega_{AB}%
\Vert\sigma_{AB}),V(\omega_{AB}\Vert\sigma_{AB}),T(\omega_{AB}\Vert\sigma
_{AB})<\infty$ and $V(\omega_{AB}\Vert\sigma_{AB})>0$ hold, such that
\eqref{eq:upper-bnd-hypo} applies and we find that%
\begin{multline}
P_{\mathcal{N}}^{\leftrightarrow}(n,\varepsilon)\leq D(\omega_{AB}\Vert
\sigma_{AB})+\sqrt{\frac{V(\omega_{AB}\Vert\sigma_{AB})}{n}}\Phi
^{-1}(\varepsilon)\label{eq:tp-simul-bound}\\
+O\!\left(  \frac{\log n}{n}\right)  .
\end{multline}
The quantities $D(\omega_{AB}\Vert\sigma_{AB})$ and $V(\omega_{AB}\Vert
\sigma_{AB})$ are finite for faithful quantum Gaussian states of finite
energy, which holds by inspecting \eqref{eq:rel-ent-G} and
\eqref{eq:rel-ent-var-G}, and in Appendix~\ref{sec:finiteness}, we argue that
the quantity $T(\omega_{AB}\Vert\sigma_{AB})$ is finite as well.

Note that both \eqref{eq:SC-thermal} and \eqref{eq:new-bound} can be derived
from \eqref{eqn:hyptest}. The point of deviation in the two derivations is
that it is possible, on the one hand, to invoke the Berry--Esseen theorem
\cite{KS10,S11} in order to arrive at \eqref{eq:new-bound}, due to the results
of \cite{LMGA17} and our arguments in Appendices~\ref{sec:AEP-hypo} and
\ref{sec:finiteness}. That is, \cite{LMGA17} showed how to perform
teleportation simulation of a single-mode phase-insensitive thermal bosonic
channel using a finite-energy resource state, and our
Appendix~\ref{sec:finiteness} argues how $T(\omega_{AB}\Vert\sigma_{AB})$ is
finite for finite-energy Gaussian states. Thus, the Berry--Esseen theorem can be
invoked as shown in Appendix~\ref{sec:AEP-hypo} and so
\eqref{eq:upper-bnd-hypo} applies. On the other hand, for the derivation of
\eqref{eq:SC-thermal}, the ideal infinite-energy Choi state of the channel is
used as the resource state, but it is not known if $T(\omega_{AB}\Vert
\sigma_{AB})$ is finite in such a scenario. Hence, unless this is proven, we
cannot invoke \eqref{eq:upper-bnd-hypo}. Therefore, other techniques, such as
the Chebyshev inequality, were used in \cite{WTB16} to arrive at~\eqref{eq:SC-thermal}.

\section{Methods}

Given the background reviewed above, we are now in a position to discuss the
main contribution of our paper. We modify the finite-energy teleportation
simulation approach of \cite{LMGA17} in the following way:\ Given a thermal
channel with $\tau=\eta$ and $y=(1-\eta)(2N_{B}+1)$, we find a finite-energy,
two-mode Gaussian state in standard form such that

\begin{enumerate}
\item it satisfies \eqref{eq:TP-simulation-1-mode-G},

\item its smaller symplectic eigenvalue is just larger than one, and

\item its suboptimal relative entropy of entanglement is equal to the
suboptimal relative entropy of entanglement of the Choi state of the channel,
the latter of which is given by \eqref{eq:rel-ent-thermal}.
\end{enumerate}

\noindent Any resource state that simulates the channel should satisfy the
first constraint. We impose the second constraint to ensure that the state we
find is a faithful Gaussian state, such that its relative entropy and relative
entropy variance to a separable Gaussian state can be easily evaluated using
the formulas in \eqref{eq:rel-ent-G} and \eqref{eq:rel-ent-var-G}.
As discussed above, the relative entropy of entanglement of the resource state
should at least be equal to that of the Choi state, in order to simulate a
channel. In order to ensure that we find a good upper bound on the
secret-key-agreement capacity, we have imposed the third constraint on
suboptimal relative entropy of entanglement.
We find these states by numerically solving the above constraints with the aid
of a computer program \footnote{Mathematica files are available in the source
files of our arXiv post.}, and we remark that finding an analytical solution
in this case appears to be far more complicated than for the case from
\cite{LMGA17}, due to the fact that the suboptimal relative entropy of
entanglement is a much more complicated function of the covariance matrix
elements. In some cases, it is possible to find multiple solutions for the
states that satisfy these constraints. For our purpose, any of these states
can be chosen. We also note that the flexibility afforded by having a
teleportation simulation with negative gain~$g$, as discussed in
Section~\ref{sec:tp-simulation}, is critical for us to solve these
constraints by numerical search. With these finite-energy states in hand, we
then numerically compute the relative entropy variance in \eqref{eq:rel-ent-var-G} and can
apply the bound in \eqref{eq:tp-simul-bound}.

\begin{figure*}[ptb]
\begin{center}
\subfloat[]{\includegraphics[width=.85\columnwidth]{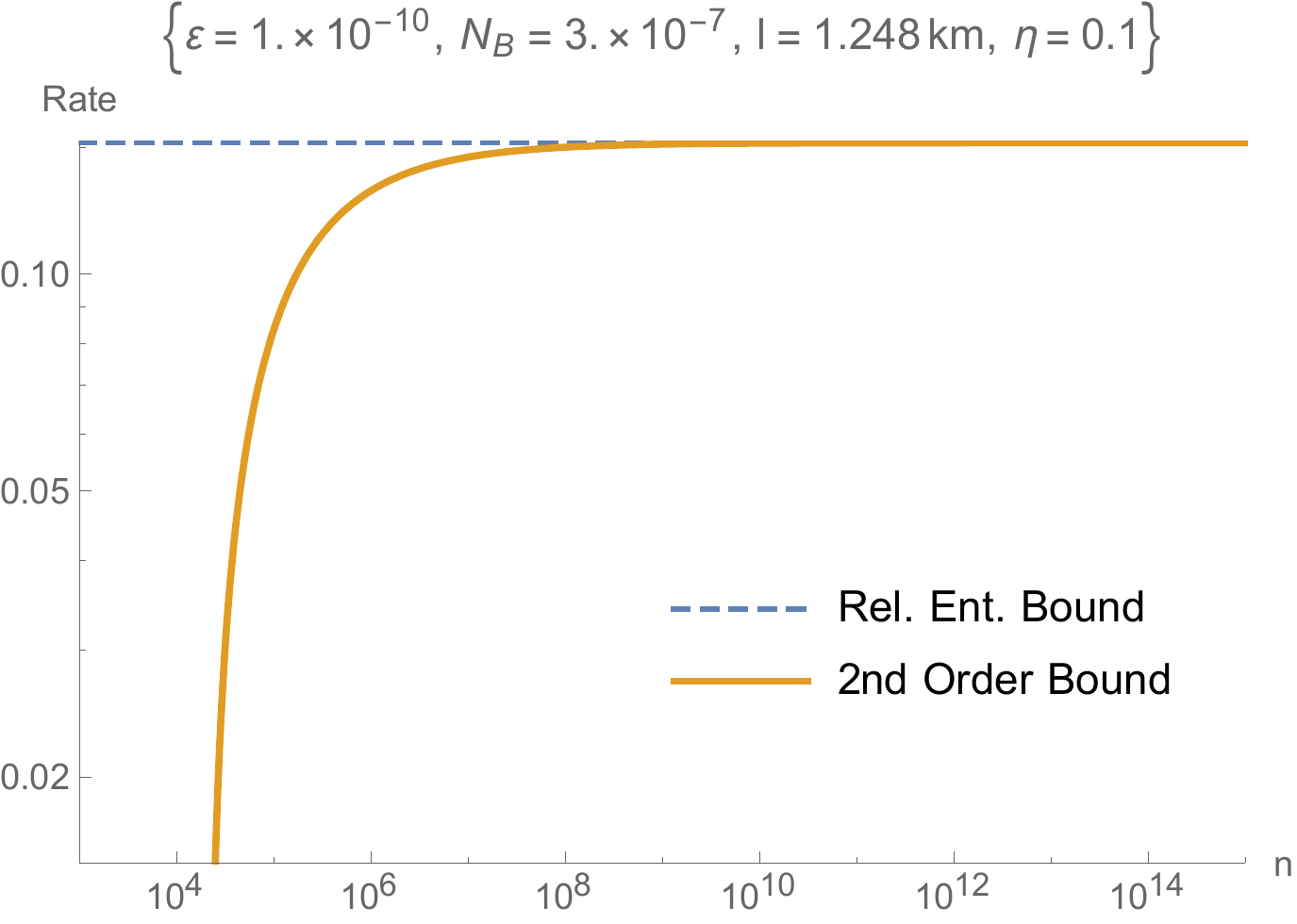}}\qquad
\qquad\subfloat[]{\includegraphics[width=.85\columnwidth]{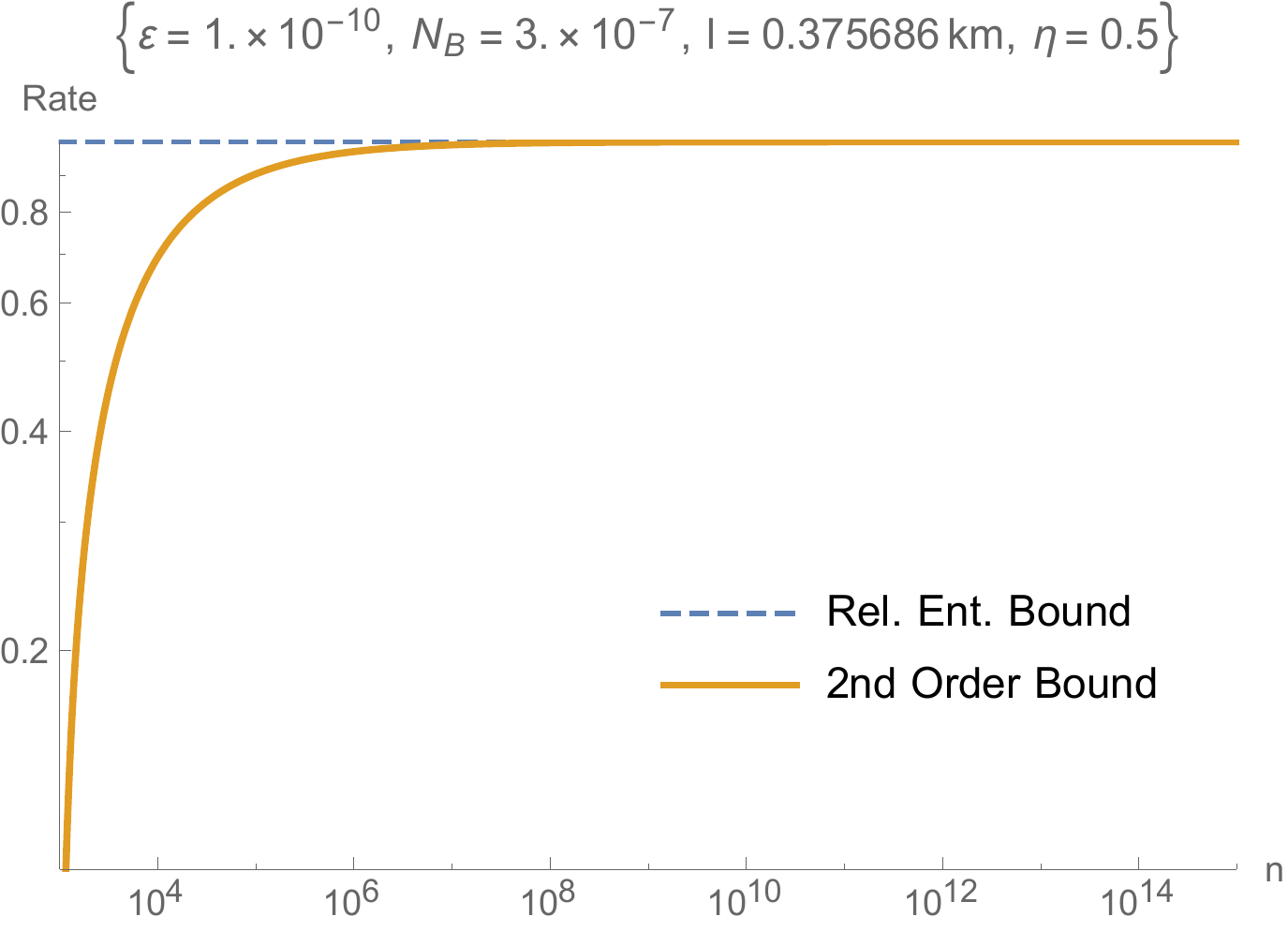}}\newline%
\subfloat[]{\includegraphics[width=.85\columnwidth]{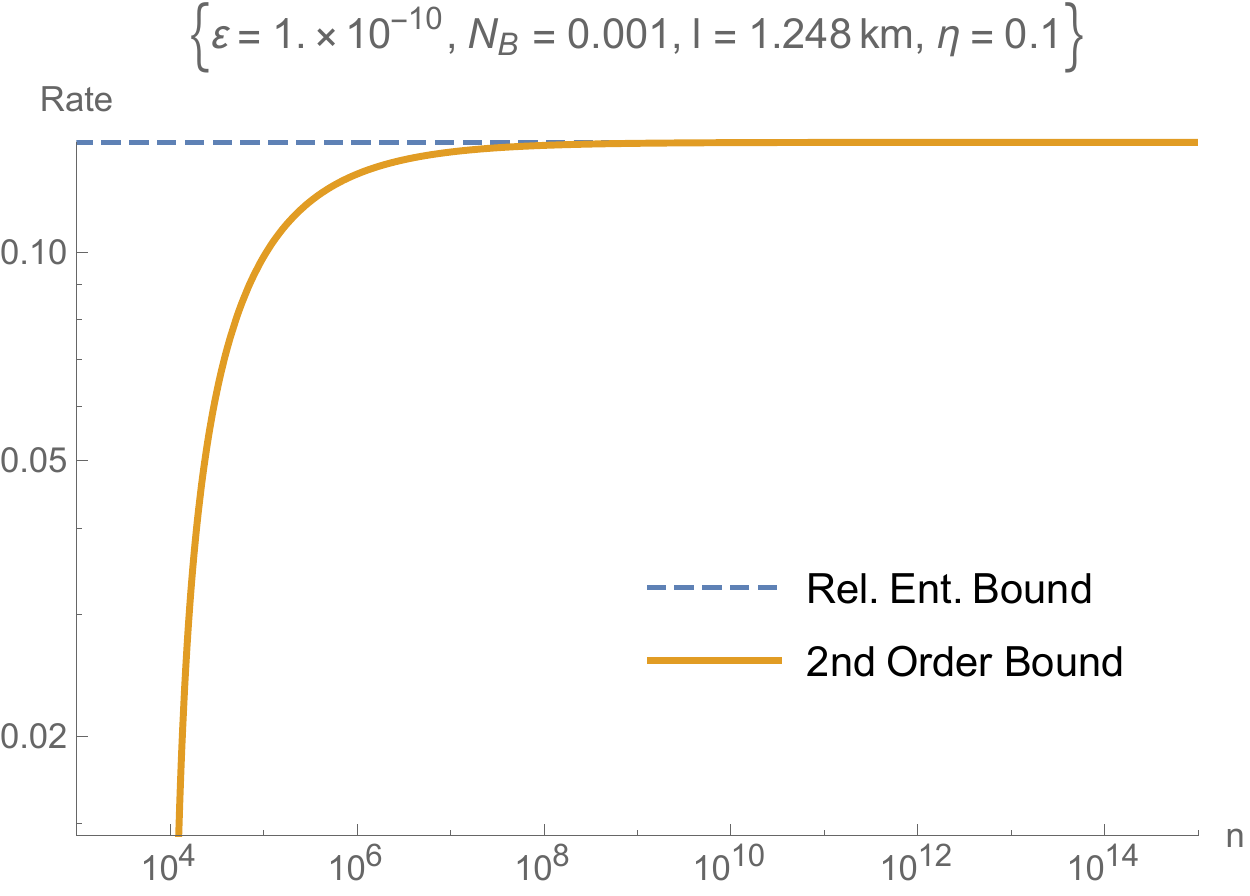}}\qquad
\qquad\subfloat[]{\includegraphics[width=.85\columnwidth]{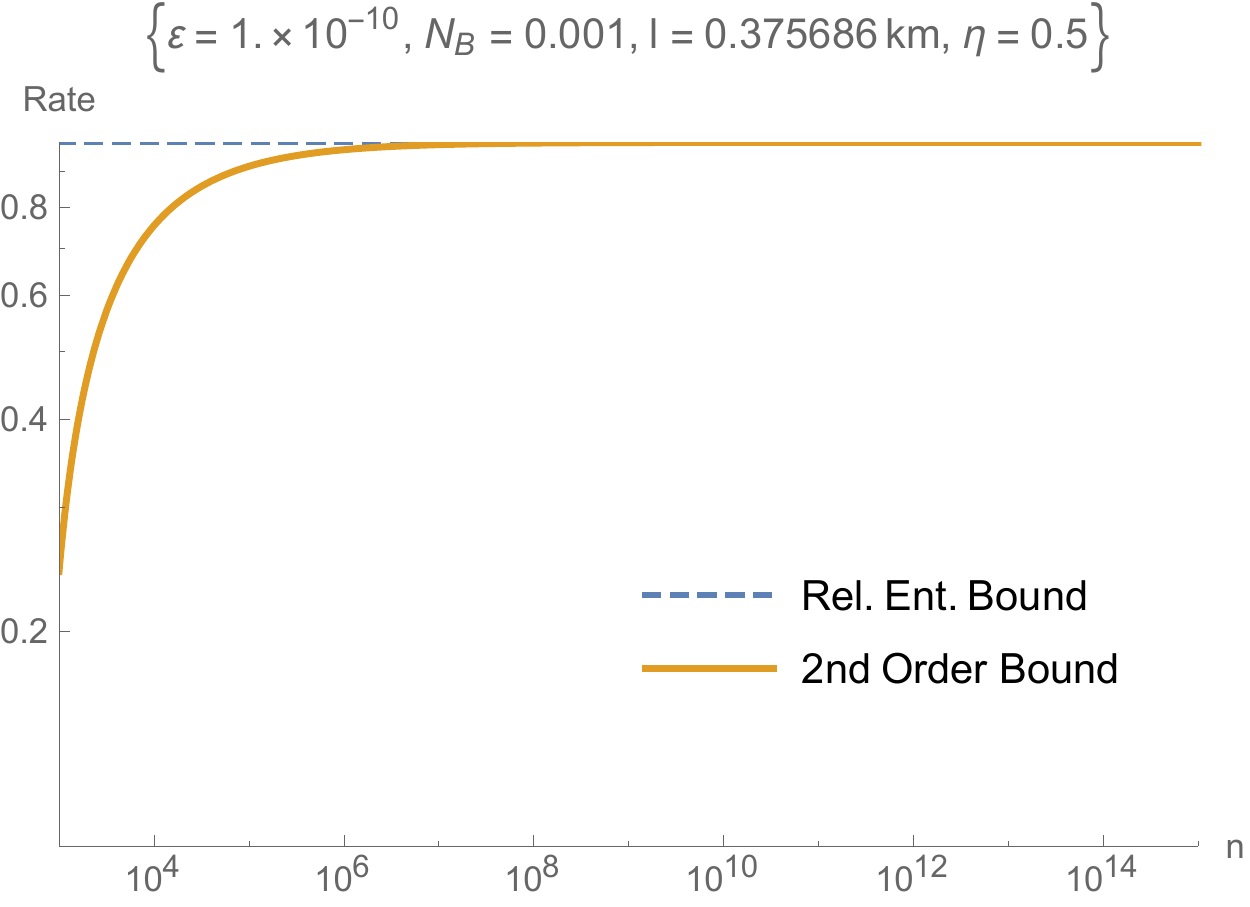}}\newline%
\subfloat[]{\includegraphics[width=.85\columnwidth]{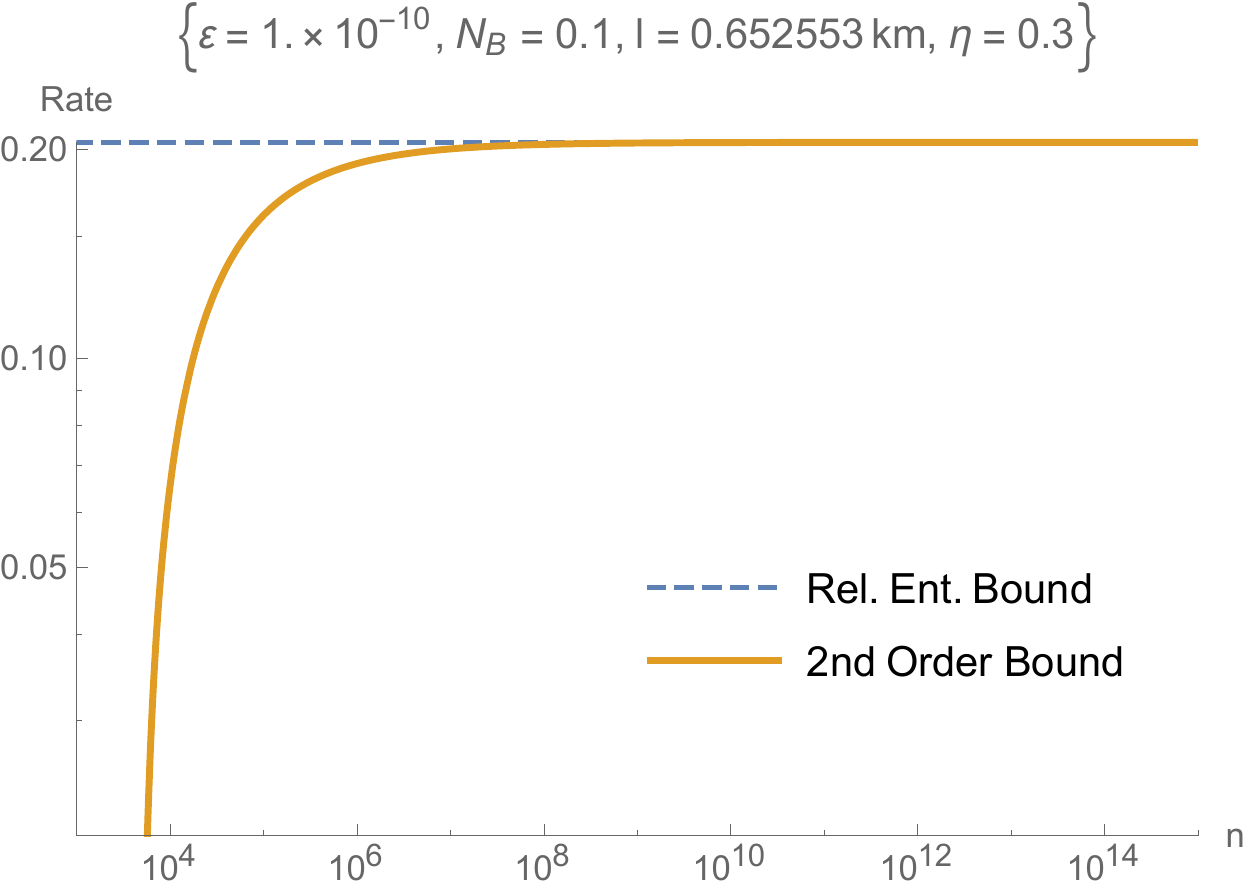}}\qquad
\qquad\subfloat[]{\includegraphics[width=.85\columnwidth]{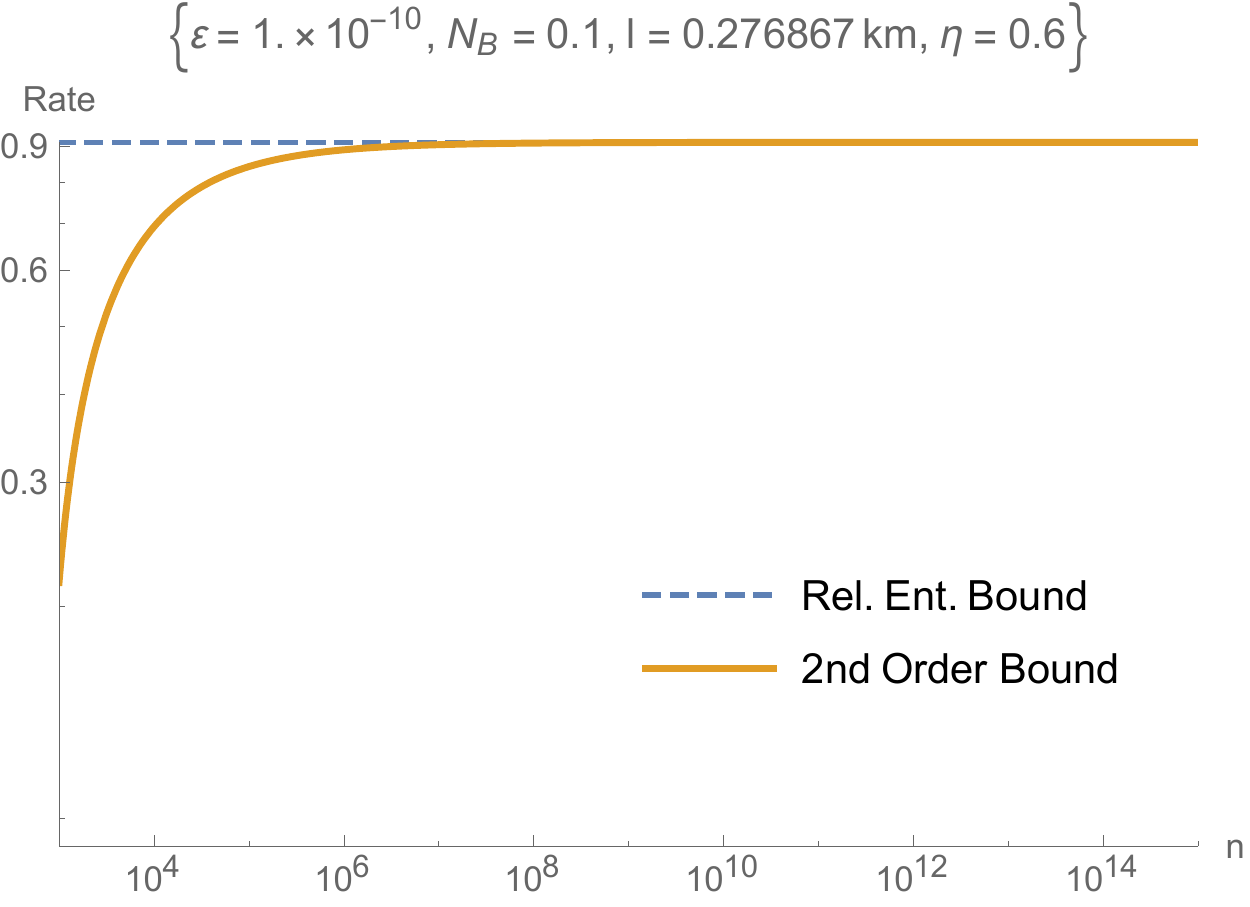}}\newline
\end{center}
\caption{The figures plot upper bounds on the non-asymptotic
secret-key-agreement capacity of the thermal channels of transmissivity
$\eta\in(0,1)$ and thermal mean photon number $N_{B}>0$, given by the
second-order approximation from \eqref{eq:new-bound}. In each figure, we
select certain values of $\eta$ (corresponding to a certain distance $L$ via
$\eta=\exp[-L/L_{0}]$) and $N_{B}$, with the choices indicated above each
figure. In all cases, we take the conservative value of $\varepsilon=10^{-10}$
as indicated in \cite{TCGR12}. Each figure indicates that the asymptotic
secret-key-agreement capacity is too pessimistic of a benchmark for
demonstrating a quantum repeater when using the channel a finite number of
times. That is, there is an appreciable difference between the asymptotic and
non-asymptotic secret-key-agreement capacity.}%
\label{fig:results}%
\end{figure*}

\section{Results}

\label{sec:results} In Figure~\ref{fig:results}, we plot upper bounds on the
asymptotic secret-key-agreement capacity of the thermal channel given by
\eqref{eq:rel-ent-thermal} (dashed line) and upper bounds on the
non-asymptotic secret-key-agreement capacity given by~\eqref{eq:new-bound}
(solid line) versus the number of channel uses. It is important to stress that
the latter bound is only an approximation (known as the normal approximation)
if $n$ is not sufficiently large (i.e., $n$ should be proportional to
$1/\varepsilon^{2}$ in order for the bounds to really apply). At the same
time, many prior works have shown that the normal approximation is an
excellent approximation for non-asymptotic capacities even for small $n$
\cite{polyanskiy10,P10,MW12,P13,TBR15}. In each case, we choose the
key-quality parameter $\varepsilon$ to be $10^{-10}$, in accordance with the
same conservative value chosen in \cite{TCGR12}. In the plots, we select
$\eta\in(0,1)$, (hence the corresponding distance $L$) and the thermal mean
photon number $N_{B}>0$ as indicated above each figure. The distance $L$ can
be related to the transmissivity $\eta$ of the thermal channel as $\eta
=\exp[-L/L_{0}]$, where $L_{0}$ is the fiber attenuation length \cite{RGRKVRHWE17}. In the plots,
we consider $L_{0}=0.542$~km \cite{RGRKVRHWE17}. The thermal mean
photon number $N_{B}$ relevant in experimental contexts, whenever thermal
noise is due exclusively to dark counts, is given by the dark counts per
second times the integration period $t_{\mathrm{int}}$. In the plots, the
lowest $N_{B}$ we consider corresponds to a dark count rate of $10$ per
second and $t_{\mathrm{int}}=30$~ns \cite[Section~VI]%
{RGRKVRHWE17}. For completeness, we also consider higher values of
$N_{B}$, which could occur due to excessive background thermal radiation or
tampering by an eavesdropper.

As noted in the introduction of our paper, these upper bounds can be
interpreted to serve as benchmarks for quantum repeaters \cite{L15}. That is,
the upper bounds on secret-key-agreement capacity hold for any protocol that
uses the channel and LOCC\ but is not allowed to use a quantum repeater. As
such, exceeding these upper bounds constitutes a demonstration of a quantum
repeater \cite{L15}. What our results indicate is that the previous upper
bounds from \cite{PLOB15,WTB16}\ on the asymptotic secret-key-agreement
capacity are too pessimistic of a benchmark for protocols that are only using
the channel a finite number of times. As such, the burden of demonstrating 
a quantum repeater is now somewhat relieved in comparison to what was
previously thought would be necessary.

From an experimental perspective, it could be of interest to perform a test
using the results of our paper in order to demonstrate a working quantum
repeater. A convincing approach for doing so would be to conduct an actual
secret-key-distillation protocol over some finite number of uses of the
channel and determine what secret-key rates can be achieved. \cite[Section~IV]%
{RGRKVRHWE17} details methods for determining secret-key rates that are
achievable in particular physical setups. For a given rate and number of
channel uses, one can then compare the results with our plots (or other plots
generated via the same method for different parameter values) to determine if the rate is achieved is larger
than the upper bounds in our plots; if it is the case, then one can claim a
working quantum repeater, albeit with the understanding that our upper bounds are
the normal approximations of the true finite-length upper bounds (as discussed previously). This approach is to be
contrasted with those that estimate the quantum bit-error rate from just a few
channel uses and then use this parameter to calculate an asymptotic key rate
(see the review in \cite{SBCDLP09} for discussions of such approaches).

\section{Conclusion}

In this paper, we showed how to extend the teleportation simulation method of
\cite{LMGA17} to the relative entropy of entanglement measure. By combining
with prior results in \cite{WTB16} regarding non-asymptotic
secret-key-agreement capacity, this extension leads to improved bounds on the
non-asymptotic secret-key-agreement capacity of a thermal bosonic channel, in
certain parameter regimes. Given that upper bounds on secret-key-agreement
capacity have been advocated as a way to assess the performance of a quantum
repeater, our results indicate that previous bounds from \cite{PLOB15,WTB16}
are too pessimistic, and it should be somewhat easier to demonstrate a working
quantum repeater in the realistic regime of a finite number of channel uses.

We remark that our approach can be extended to quantum amplifier channels, but
we did not discuss these channels in any detail because they appear to be most
prominently physically relevant in exotic relativistic communication scenarios
\cite{bradler2012quantum,bradler2015black,QW16}. Our approach also applies to single-mode additive-noise Gaussian channels.

Going forward from here, it would be interesting to generalize our results to
multimode bosonic communication channels \cite{CEGH08} or channels that are
not phase-insensitive. As discussed previously \cite{TSW16,TSW17,WTB16}, it
would also be good to determine bounds on performance when there is an average
energy constraint at the input of each channel use. One should expect to find
improved upper bounds due to this extra constraint.

\begin{acknowledgments}
We are grateful to Gerardo Adesso, Boulat Bash, Mario Berta, Zachary Dutton,
Jens Eisert, Saikat Guha, Jonathan P. Dowling, Jeffrey H. Shapiro, and Marco
Tomamichel for discussions. We also thank the anonymous referees for their constructive comments that helped to improve our paper. We acknowledge support from the Office of Naval\ Research.
\end{acknowledgments}

\appendix

\section{Little room for improving the strong converse bound in
\eqref{eq:SC-bound}}

\label{sec:disp-rev-coh-info}Here we argue why we think it will not be
possible to improve upon the upper bound in \eqref{eq:SC-bound}, up to
lower-order terms. Before proceeding, recall that the conditional quantum
entropy and conditional entropy variance \cite{TH12} are\ defined for a
bipartite state $\rho_{AB}$ as%
\begin{align}
H(A|B)_{\rho}  &  \equiv-D(\rho_{AB}\Vert I_{A}\otimes\rho_{B}),\\
V(A|B)_{\rho}  &  \equiv V(\rho_{AB}\Vert I_{A}\otimes\rho_{B}).
\end{align}
The coherent information is defined as $I(A\rangle B)_{\rho}\equiv
-H(A|B)_{\rho}$ \cite{PhysRevA.54.2629}\ and its corresponding variance is
$V(A\rangle B)_{\rho}\equiv V(A|B)_{\rho}$. In \cite[Section~6.2]{WTB16}, the
following achievability bound was established for $\mathcal{N}$ a
finite-dimensional channel:%
\begin{equation}
P_{\mathcal{N}}^{\leftrightarrow}(n,\varepsilon)\geq I_{\operatorname{rev}%
}(\mathcal{N})+\sqrt{\frac{V_{\operatorname{rev}}^{\varepsilon}(\mathcal{N)}%
}{n}}\Phi^{-1}(\varepsilon)+O\!\left(  \frac{\log n}{n}\right)
\label{eq:rev-low-bound}%
\end{equation}
where $I_{\operatorname{rev}}(\mathcal{N})$ is the following quantity
\cite[Section~5.3]{DJKR06} (sometimes called the channel's reverse coherent
information):%
\begin{equation}
I_{\operatorname{rev}}(\mathcal{N})\equiv\max_{|\psi\rangle_{AA^{\prime}}%
\in\mathcal{H}_{AA^{\prime}}}I(B\rangle A)_{\theta}, \label{eq:rev-coh-info}%
\end{equation}
$\theta_{AB}\equiv\mathcal{N}_{A^{\prime}\rightarrow B}(\psi_{AA^{\prime}})$,
and $V_{\operatorname{rev}}^{\varepsilon}(\mathcal{N)}$ is the channel's
reverse conditional entropy variance:%
\begin{equation}
V_{\operatorname{rev}}^{\varepsilon}(\mathcal{N)\equiv}\left\{
\begin{array}
[c]{cc}%
\min_{\psi_{AA^{\prime}}\in\Pi_{\operatorname{rev}}}V(B\rangle A)_{\theta} &
\text{for }\varepsilon<1/2\\
\max_{\psi_{AA^{\prime}}\in\Pi_{\operatorname{rev}}}V(B\rangle A)_{\theta} &
\text{for }\varepsilon\geq1/2
\end{array}
\right.  .
\end{equation}
The set $\Pi_{\operatorname{rev}}\subseteq\mathcal{D}(\mathcal{H}_{AA^{\prime
}})$ is the set of all states achieving the maximum in \eqref{eq:rev-coh-info}.

The inequality in \eqref{eq:rev-low-bound} follows from a one-shot coding
theorem \cite[Proposition~21]{WTB16}, followed by an expansion of the
hypothesis testing relative entropy as \cite{li12,TH12}%
\begin{multline}
D_{H}^{\varepsilon}(\rho^{\otimes n}\Vert\sigma^{\otimes n})\geq\\
nD(\rho\Vert\sigma)+\sqrt{nV(\rho\Vert\sigma)}\Phi^{-1}(\varepsilon)+O(\log
n).
\end{multline}
A critical step employed in the above expansion is the Berry--Esseen theorem
\cite{KS10,S11}. Rather than employing the Berry--Esseen theorem, we can
modify the proof of Theorem~2\ in \cite{li12}\ (therein instead picking
$L_{n}=\exp( nD(\rho\Vert\sigma)-\sqrt{n V(\rho\Vert\sigma)/\varepsilon}) $)
to employ the Chebyshev inequality and instead find the following expansion:%
\begin{equation}
D_{H}^{\varepsilon}(\rho^{\otimes n}\Vert\sigma^{\otimes n})\geq nD(\rho
\Vert\sigma)-\sqrt{\frac{nV(\rho\Vert\sigma)}{\varepsilon}}.
\label{eq:chebyshev-expand}%
\end{equation}
For these theorems to hold in separable infinite-dimensional Hilbert spaces,
it remains to show how to connect the coding theorem in \cite[Proposition~21]%
{WTB16} to the inequality in \eqref{eq:chebyshev-expand}, but we strongly
suspect that this should be possible. If everything holds, we would obtain the
following achievability theorem for an infinite-dimensional channel
$\mathcal{N}$:%
\begin{equation}
P_{\mathcal{N}}^{\leftrightarrow}(n,\varepsilon)\geq I_{\operatorname{rev}%
}(\mathcal{N})-\sqrt{\frac{V_{\operatorname{rev}}^{\varepsilon}(\mathcal{N)}%
}{n\varepsilon}}+O\!\left(  \frac{1}{n}\right)  .
\end{equation}
The above would hold for all finite-energy two-mode, squeezed vacuum states
passed through the channel, and one could then take a limit as the photon
number approaches infinity. The term $I_{\operatorname{rev}}(\mathcal{N})$
converges to $-\log_{2}(1-\eta)$ \cite{PLOB15}. Below we show that the
relative entropy variance $V_{\operatorname{rev}}^{\varepsilon}(\mathcal{N)}$
term converges to zero. This would then give the following bound
\begin{equation}
P_{\mathcal{L}_{\eta}}^{\leftrightarrow}(n,\varepsilon)\geq-\log_{2}%
(1-\eta)+O\!\left(  \frac{1}{n}\right)  ,
\end{equation}
leading us to our conclusion that there is little room for improving the upper
bound in \eqref{eq:SC-bound}. We stress that this remains to be worked out in detail.

We now evaluate the variance for the reverse coherent information when sending
in a two-mode squeezed vacuum to a pure-loss channel of transmissivity
$\eta\in(0,1)$. Recall that the quantity of interest is%
\begin{align}
&  V(B\rangle A)\nonumber\\
&  =\operatorname{Tr}\{\rho_{AB}\left[  \log\rho_{AB}-\log\rho_{A}\right]
^{2}\}\\
&  \qquad-\left[  H(AB)_{\rho}-H(A)_{\rho}\right]  ^{2}\\
&  =\operatorname{Tr}\{\rho_{AB}\left[  \log\rho_{AB}\right]  ^{2}%
\}\nonumber\\
&  \qquad-2\operatorname{Tr}\{\rho_{AB}\log\rho_{AB}\log\rho_{A}%
\}+\operatorname{Tr}\{\rho_{AB}\left[  \log\rho_{A}\right]  ^{2}\}\nonumber\\
&  \qquad-\left[  H(AB)_{\rho}-H(A)_{\rho}\right]  ^{2}\\
&  =\operatorname{Tr}\{\rho_{AB}\left[  \log\rho_{AB}\right]  ^{2}%
\}-H(AB)_{\rho}^{2}\nonumber\\
&  \qquad-2\left[  \operatorname{Tr}\{\rho_{AB}\log\rho_{AB}\log\rho
_{A}\}-H(A)_{\rho}H(AB)_{\rho}\right] \nonumber\\
&  \qquad+\operatorname{Tr}\{\rho_{A}\left[  \log\rho_{A}\right]
^{2}\}-H(A)_{\rho}^{2}.
\end{align}
The first and last terms we can evaluate easily using the following formula
for the entropy variance of a thermal state with mean photon number $N_{S}$
\cite[Appendix~A]{WRG15}:%
\begin{equation}
V(N_{S})=N_{S}\left(  N_{S}+1\right)  \left[  \log\left(  1+\frac{1}{N_{S}%
}\right)  \right]  ^{2}.
\end{equation}
For the first, using the notion of purification, purifying with $\psi_{ABE}$,
and observing that $\psi_{E}$ is a thermal state with mean photon number
$\left(  1-\eta\right)  N_{S}$, we find that%
\begin{align}
&  \operatorname{Tr}\{\rho_{AB}\left[  \log\rho_{AB}\right]  ^{2}%
\}-H(AB)_{\rho}^{2}\nonumber\\
&  =\operatorname{Tr}\{\psi_{E}\left[  \log\psi_{E}\right]  ^{2}\}-H(E)_{\psi
}^{2}\\
&  =\left(  1-\eta\right)  N_{S}\left(  \left(  1-\eta\right)  N_{S}+1\right)
\left[  \log\left(  1+\frac{1}{\left(  1-\eta\right)  N_{S}}\right)  \right]
^{2}.
\end{align}
For the last term, we observe that $\rho_{A}$ is a thermal state with mean
photon number $N_{S}$, which implies that%
\begin{multline}
\operatorname{Tr}\{\rho_{A}\left[  \log\rho_{A}\right]  ^{2}\}-H(A)_{\rho}%
^{2}\\
=N_{S}\left(  N_{S}+1\right)  \left[  \log\left(  1+\frac{1}{N_{S}}\right)
\right]  ^{2}.
\end{multline}
So it remains to handle the middle term. Consider that%
\begin{align}
&  \operatorname{Tr}\{\rho_{AB}\log\rho_{AB}\log\rho_{A}\}\nonumber\\
&  =\operatorname{Tr}\{\psi_{ABE}\log\rho_{AB}\log\rho_{A}\}\\
&  =\operatorname{Tr}\{\psi_{ABE}\log\psi_{E}\log\rho_{A}\}\\
&  =\operatorname{Tr}\{\psi_{AE}\log\psi_{E}\log\rho_{A}\}.
\end{align}
Consider that we can write%
\begin{align}
\psi_{E}  &  =\left[  \left(  1-\eta\right)  N_{S}+1\right]  ^{-1}\left(
1+\frac{1}{\left(  1-\eta\right)  N_{S}}\right)  ^{-\hat{n}_{E}},\\
\rho_{A}  &  =\left[  N_{S}+1\right]  ^{-1}\left(  1+\frac{1}{N_{S}}\right)
^{-\hat{n}_{A}},
\end{align}
where $\hat{n}_{E}$ and $\hat{n}_{A}$ are the number operators. This means
that%
\begin{align}
&  \operatorname{Tr}\{\psi_{AE}\log\psi_{E}\log\rho_{A}\}\nonumber\\
&  =\operatorname{Tr}\Bigg\{\psi_{AE}\log\left[  \left[  \left(
1-\eta\right)  N_{S}+1\right]  ^{-1}\left(  1+\frac{1}{\left(  1-\eta\right)
N_{S}}\right)  ^{-\hat{n}_{E}}\right] \nonumber\\
&  \qquad\times\log\left[  \left[  N_{S}+1\right]  ^{-1}\left(  1+\frac
{1}{N_{S}}\right)  ^{-\hat{n}_{A}}\right]  \Bigg\}\\
&  =\operatorname{Tr}\Bigg\{\psi_{AE}\log\left[  \left[  \left(
1-\eta\right)  N_{S}+1\right]  \left(  1+\frac{1}{\left(  1-\eta\right)
N_{S}}\right)  ^{\hat{n}_{E}}\right] \nonumber\\
&  \qquad\times\log\left[  \left[  N_{S}+1\right]  \left(  1+\frac{1}{N_{S}%
}\right)  ^{\hat{n}_{A}}\right]  \Bigg\}
\end{align}%
\begin{multline}
=\log\left[  \left(  1-\eta\right)  N_{S}+1\right]  \log\left[  N_{S}+1\right]
\\
+\log\left[  \left[  \left(  1-\eta\right)  N_{S}+1\right]  \right]
\log\left[  \left(  1+\frac{1}{N_{S}}\right)  \right]  \operatorname{Tr}%
\left\{  \psi_{AE}\hat{n}_{A}\right\} \\
+\log\left[  1+\frac{1}{\left(  1-\eta\right)  N_{S}}\right]  \log\left[
N_{S}+1\right]  \operatorname{Tr}\left\{  \psi_{AE}\hat{n}_{E}\right\} \\
+\log\left[  1+\frac{1}{\left(  1-\eta\right)  N_{S}}\right]  \log\left[
1+\frac{1}{N_{S}}\right]  \operatorname{Tr}\left\{  \psi_{AE}\left(  \hat
{n}_{A}\otimes\hat{n}_{E}\right)  \right\}
\end{multline}%
\begin{multline}
=\log\left[  \left(  1-\eta\right)  N_{S}+1\right]  \log\left[  N_{S}+1\right]
\\
+N_{S}\log\left[  \left[  \left(  1-\eta\right)  N_{S}+1\right]  \right]
\log\left[  \left(  1+\frac{1}{N_{S}}\right)  \right] \\
+\left(  1-\eta\right)  N_{S}\log\left[  1+\frac{1}{\left(  1-\eta\right)
N_{S}}\right]  \log\left[  N_{S}+1\right] \\
+\log\left[  1+\frac{1}{\left(  1-\eta\right)  N_{S}}\right]  \log\left[
1+\frac{1}{N_{S}}\right]\times\\  \operatorname{Tr}\left\{  \psi_{AE}\left(  \hat
{n}_{A}\otimes\hat{n}_{E}\right)  \right\}  .
\end{multline}
We note that the third equality follows by applying the identity $\log(a
b^{\hat{x}}) = \log(a) + \hat{x}\log(b)$ for positive scalars $a$ and $b$ and
a positive operator $\hat{x}$. So we need to evaluate the term
$\operatorname{Tr}\left\{  \psi_{AE}\left(  \hat{n}_{A}\otimes\hat{n}%
_{E}\right)  \right\}  $. Consider that sending a number state $|n\rangle
\langle n|$ through a beamsplitter of transmissivity $1-\eta$ leads to the
following transformation:%
\begin{equation}
|n\rangle\langle n|_{A^{\prime}}\rightarrow\sum_{k=0}^{n}\binom{n}{k}\left(
1-\eta\right)  ^{k}\eta^{n-k}|k\rangle\langle k|_{E}.
\end{equation}
The two-mode squeezed vacuum at the input has the following form:%
\begin{equation}
\frac{1}{\sqrt{N_{S}+1}}\sum_{n=0}^{\infty}\sqrt{\left(  \frac{N_{S}}{N_{S}%
+1}\right)  ^{n}}|n\rangle_{A}|n\rangle_{A^{\prime}}.
\end{equation}
However since we are evaluating $\operatorname{Tr}\left\{  \psi_{AE}\left(
\hat{n}_{A}\otimes\hat{n}_{E}\right)  \right\}  $, and $\hat{n}_{A}$ and
$\hat{n}_{E}$ are diagonal in the number basis, this is equivalent to the
following:%
\begin{align}
&  \frac{1}{N_{S}+1}\sum_{n=0}^{\infty}\sum_{k=0}^{n}\left(  \frac{N_{S}%
}{N_{S}+1}\right)  ^{n}\binom{n}{k}\left(  1-\eta\right)  ^{k}\eta^{n-k}%
\times\nonumber\\
&  \qquad\operatorname{Tr}\{\left(  |n\rangle\langle n|_{A}\otimes
|k\rangle\langle k|\right)  \left(  \hat{n}_{A}\otimes\hat{n}_{E}\right)
\}\nonumber\\
&  =\frac{1}{N_{S}+1}\sum_{n=0}^{\infty}\sum_{k=0}^{n}\left(  \frac{N_{S}%
}{N_{S}+1}\right)  ^{n}\binom{n}{k}\left(  1-\eta\right)  ^{k}\eta^{n-k}nk\\
&  =\frac{1}{N_{S}+1}\sum_{n=0}^{\infty}n\left(  \frac{N_{S}}{N_{S}+1}\right)
^{n}\sum_{k=0}^{n}\binom{n}{k}\left(  1-\eta\right)  ^{k}\eta^{n-k}k.
\end{align}
Consider that the expression $\sum_{k=0}^{n}\binom{n}{k}\left(  1-\eta\right)
^{k}\eta^{n-k}k$ is equal to the mean of a binomial random variable with
parameter $1-\eta$, and so%
\begin{equation}
\sum_{k=0}^{n}\binom{n}{k}\left(  1-\eta\right)  ^{k}\eta^{n-k}k=n\left(
1-\eta\right)  ,
\end{equation}
implying that the last line above is equal to%
\begin{equation}
\left(  1-\eta\right)  \frac{1}{N_{S}+1}\sum_{n=0}^{\infty}n^{2}\left(
\frac{N_{S}}{N_{S}+1}\right)  ^{n}.
\end{equation}
This is then equal to the second moment of a geometric random variable with
parameter $p=1/\left(  N_{S}+1\right)  $, so that%
\begin{align}
&  \left(  1-\eta\right)  \frac{1}{N_{S}+1}\sum_{n=0}^{\infty}n^{2}\left(
\frac{N_{S}}{N_{S}+1}\right)  ^{n}\nonumber\\
&  =\left(  1-\eta\right)  \left(  N_{S}\left(  N_{S}+1\right)  +N_{S}%
^{2}\right) \\
&  =\left(  1-\eta\right)  N_{S}\left(  2N_{S}+1\right)  .
\end{align}
Plugging into the above, we find the reduction%
\begin{multline}
=\log\left[  \left(  1-\eta\right)  N_{S}+1\right]  \log\left[  N_{S}+1\right]
\\
+N_{S}\log\left[  \left[  \left(  1-\eta\right)  N_{S}+1\right]  \right]
\log\left[  \left(  1+\frac{1}{N_{S}}\right)  \right] \\
+\left(  1-\eta\right)  N_{S}\log\left[  1+\frac{1}{\left(  1-\eta\right)
N_{S}}\right]  \log\left[  N_{S}+1\right] \\
+\left(  1-\eta\right)  N_{S}\left(  2N_{S}+1\right)  \log\left[  1+\frac
{1}{\left(  1-\eta\right)  N_{S}}\right]  \log\left[  1+\frac{1}{N_{S}%
}\right]  .
\end{multline}
From this we should subtract the following quantity%
\begin{align}
&  H(A)_{\rho}H(AB)_{\rho}\nonumber\\
&  =g(N_{S})g(\left(  1-\eta\right)  N_{S})\\
&  =\left[  \left(  N_{S}+1\right)  \log\left(  N_{S}+1\right)  -N_{S}\log
N_{S}\right]  \times\nonumber\\
&  \left[
\begin{array}
[c]{c}%
\left(  \left(  1-\eta\right)  N_{S}+1\right)  \log\left(  \left(
1-\eta\right)  N_{S}+1\right) \\
-\left(  1-\eta\right)  N_{S}\log\left(  1-\eta\right)  N_{S}%
\end{array}
\right] \\
&  =\left[  N_{S}\log\left(  1+\frac{1}{N_{S}}\right)  +\log\left(
N_{S}+1\right)  \right]  \times\nonumber\\
&  \left[  \left(  1-\eta\right)  N_{S}\log\left(  1+\frac{1}{\left(
1-\eta\right)  N_{S}}\right)  +\log\left(  \left(  1-\eta\right)
N_{S}+1\right)  \right]
\end{align}%
\begin{multline}
=\left(  1-\eta\right)  N_{S}^{2}\log\left(  1+\frac{1}{N_{S}}\right)
\log\left(  1+\frac{1}{\left(  1-\eta\right)  N_{S}}\right) \\
+\left(  1-\eta\right)  N_{S}\log\left(  N_{S}+1\right)  \log\left(
1+\frac{1}{\left(  1-\eta\right)  N_{S}}\right) \\
+N_{S}\log\left(  1+\frac{1}{N_{S}}\right)  \log\left(  \left(  1-\eta\right)
N_{S}+1\right) \\
+\log\left(  N_{S}+1\right)  \log\left(  \left(  1-\eta\right)  N_{S}%
+1\right)  ,
\end{multline}
leading to%
\begin{align}
&  \operatorname{Tr}\{\rho_{AB}\log\rho_{AB}\log\rho_{A}\}-H(A)_{\rho
}H(AB)_{\rho}\nonumber\\
&  =\left(  1-\eta\right)  N_{S}\left(  2N_{S}+1\right)  \log\left[
1+\frac{1}{\left(  1-\eta\right)  N_{S}}\right]  \log\left[  1+\frac{1}{N_{S}%
}\right] \nonumber\\
&  \qquad-\left(  1-\eta\right)  N_{S}^{2}\log\left(  1+\frac{1}{N_{S}%
}\right)  \log\left(  1+\frac{1}{\left(  1-\eta\right)  N_{S}}\right) \\
&  =\left(  1-\eta\right)  N_{S}\left(  N_{S}+1\right)  \log\left[  1+\frac
{1}{\left(  1-\eta\right)  N_{S}}\right]  \log\left[  1+\frac{1}{N_{S}%
}\right]  .
\end{align}
Putting everything together, we find that the variance of the reverse coherent
information is given by%
\begin{multline}
\left(  1-\eta\right)  N_{S}\left(  \left(  1-\eta\right)  N_{S}+1\right)
\left[  \log\left(  1+\frac{1}{\left(  1-\eta\right)  N_{S}}\right)  \right]
^{2}\\
-2\left(  1-\eta\right)  N_{S}\left(  N_{S}+1\right)  \log\left[  1+\frac
{1}{\left(  1-\eta\right)  N_{S}}\right]  \log\left[  1+\frac{1}{N_{S}}\right]
\\
+N_{S}\left(  N_{S}+1\right)  \left[  \log\left(  1+\frac{1}{N_{S}}\right)
\right]  ^{2}.
\end{multline}
For large $N_{S}$, we have that $\left(  \left(  1-\eta\right)  N_{S}%
+1\right)  \approx\left(  1-\eta\right)  N_{S}$ and $\left(  N_{S}+1\right)
\approx N_{S}$, so that the above reduces to%
\begin{equation}
\approx\left[  \left(  1-\eta\right)  N_{S}\log\left(  1+\frac{1}{\left(
1-\eta\right)  N_{S}}\right)  -N_{S}\log\left(  1+\frac{1}{N_{S}}\right)
\right]  ^{2},
\end{equation}
which converges to zero as $N_{S}\rightarrow\infty$.

\section{Weak converse bounds for secret-key-agreement capacities}

\label{sec:weak-converse-bnds}Here we argue for the weak-converse bounds given
in \eqref{eq:WC-bound} and \eqref{eq:WC-thermal}, and even more general
weak-converse bounds. The weak-converse bounds are a direct consequence of the
bounds in \cite{WTB16} and \cite[Eq.~(2)]{WR12} (see also \cite[Eq.~(134)]%
{MW12}).

First, recall from \cite[Eq.~(2)]{WR12} and \cite[Eq.~(134)]{MW12} that the
following bound holds for hypothesis testing relative entropy for
$\varepsilon\in(0,1)$:%
\begin{equation}
D_{H}^{\varepsilon}(\rho\Vert\sigma)\leq\frac{1}{1-\varepsilon}\left[
D(\rho\Vert\sigma)+h_{2}(\varepsilon)\right]  . \label{eq:WR-bound-to-rel-ent}%
\end{equation}
To see this, consider that the definition of $D_{H}^{\varepsilon}(\rho
\Vert\sigma)$ can be further constrained as%
\begin{multline}
D_{H}^{\varepsilon}(\rho\Vert\sigma)=\label{eq:hypo-rewrite}\\
-\log_{2}\inf_{\Lambda}\{\operatorname{Tr}\{\Lambda\sigma\}:0\leq\Lambda\leq
I\wedge\operatorname{Tr}\{\Lambda\rho\}=1-\varepsilon\}.
\end{multline}
That is, it suffices to optimize over measurement operators that meet the
constraint $\operatorname{Tr}\{\Lambda\rho\}\geq1-\varepsilon$ with equality.
This follows because for any measurement operator $\Lambda$ such that
$\operatorname{Tr}\{\Lambda\rho\}>1-\varepsilon$, we can modify it by scaling
it by a positive number $\lambda\in(0,1)$ such that $\operatorname{Tr}%
\{\left(  \lambda\Lambda\right)  \rho\}=1-\varepsilon$. The new operator
$\lambda\Lambda$ is a legitimate measurement operator and the error
probability $\operatorname{Tr}\{(\lambda\Lambda)\sigma\}$ only decreases under
this scaling (i.e., $\operatorname{Tr}\{(\lambda\Lambda)\sigma
\}<\operatorname{Tr}\{\Lambda\sigma\}$), which allows us to conclude
\eqref{eq:hypo-rewrite}. Now for any measurement operator $\Lambda$ such that
$\operatorname{Tr}\{\Lambda\rho\}=1-\varepsilon$, the monotonicity of quantum
relative entropy \cite{Lindblad1975} with respect to quantum channels implies
that%
\begin{align}
&  D(\rho\Vert\sigma)\nonumber\\
&  \geq D(\{1-\varepsilon,\varepsilon\}\Vert\{\operatorname{Tr}\{\Lambda
\sigma\},1-\operatorname{Tr}\{\Lambda\sigma\}\})\\
&  =(1-\varepsilon)\log_{2}\left(  \frac{1-\varepsilon}{\operatorname{Tr}%
\{\Lambda\sigma\}}\right)  +\varepsilon\log_{2}\!\left(  \frac{\varepsilon
}{1-\operatorname{Tr}\{\Lambda\sigma\}}\right) \\
&  =-\left(  1-\varepsilon\right)  \log_{2}\operatorname{Tr}\{\Lambda
\sigma\}-h_{2}(\varepsilon)\nonumber\\
&  \qquad\qquad+\varepsilon\log_{2}\!\left(  \frac{1}{1-\operatorname{Tr}%
\{\Lambda\sigma\}}\right) \\
&  \geq-\left(  1-\varepsilon\right)  \log_{2}\operatorname{Tr}\{\Lambda
\sigma\}-h_{2}(\varepsilon).
\end{align}
Rewriting this gives%
\begin{equation}
-\log\operatorname{Tr}\{\Lambda\sigma\}\leq\frac{1}{1-\varepsilon}\left[
D(\rho\Vert\sigma)+h_{2}(\varepsilon)\right]  .
\end{equation}
Since this bound holds for all measurement operators $\Lambda$ satisfying
$\operatorname{Tr}\{\Lambda\rho\}=1-\varepsilon$, we can conclude \eqref{eq:WR-bound-to-rel-ent}.

To conclude the desired weak-converse bounds, we then invoke the above and
\cite[Eq.~(4.34)]{WTB16} to get that the following bound holds for any
teleportation simulable channel with associated resource state $\omega_{AB}$:%
\begin{align}
P_{\mathcal{N}}^{\leftrightarrow}(n,\varepsilon)  &  \leq\frac{1}{n}%
E_{R}^{\varepsilon}(A^{n};B^{n})_{\omega^{\otimes n}}\\
&  \leq\frac{1}{n(1-\varepsilon)}\left[  E_{R}(A^{n};B^{n})_{\omega^{\otimes
n}}+h_{2}(\varepsilon)\right] \\
&  \leq\frac{1}{(1-\varepsilon)}\left[  E_{R}(A;B)_{\omega}+\frac
{h_{2}(\varepsilon)}{n}\right]  .
\end{align}
If the channel requires an infinite-energy resource state to become
teleportation simulable, then one must take care as in the case of the proofs
in \cite[Section~8]{WTB16}, and then one finally arrives at the weak-converse
bounds in \eqref{eq:WC-bound} and~\eqref{eq:WC-thermal}.

\section{Asymptotic equipartition property for hypothesis testing relative
entropy}

\label{sec:AEP-hypo}

In this appendix, we prove that the inequality in
\eqref{eq:upper-bnd-hypo}\ holds whenever the states $\rho$ and $\sigma$
involved act on a separable Hilbert space. Here we take the convention, for
convenience, that all logarithms are with respect to the natural base, but we
note that the bound \eqref{eq:hypo-expand} applies equally well for the binary
logarithm just by rescaling.

The following proposition is available as \cite[Eq.~(6.5)]{JOPS12} and
restated as \cite[Corollary~2]{DPR15}:

\begin{proposition}
[{\cite[Eq.~(6.5)]{JOPS12}}]\label{prop:jaksic}Let $\rho$ and $\sigma$ be
faithful states acting on a separable Hilbert space $\mathcal{H}$, let
$\Lambda$ be a measurement operator acting on $\mathcal{H}$ and such that
$0\leq\Lambda\leq I$, and let $v,\theta\in\mathbb{R}$. Then%
\begin{equation}
e^{-\theta}\operatorname{Tr}\{(I-\Lambda)\rho\}+\operatorname{Tr}%
\{\Lambda\sigma\}\geq\frac{e^{-\theta}}{1+e^{v-\theta}}\Pr\{X\leq v\},
\end{equation}
where $X$ is a random variable taking values $\log(\lambda_{x}/\mu_{y})$ with
probability $\left\vert \left\langle \psi_{y}|\phi_{x}\right\rangle
\right\vert ^{2}\lambda_{x}$, where these quantities are defined in
\eqref{eq:spec-decomp-rho} and \eqref{eq:spec-decomp-sigma}.
\end{proposition}

The following proposition is based on ideas given in~\cite{DPR15}:

\begin{proposition}
Let $\rho$ and $\sigma$ be faithful states acting on a separable Hilbert space
$\mathcal{H}$, such that
\begin{align}
D(\rho\Vert\sigma),\ V(\rho\Vert\sigma),\ T(\rho\Vert\sigma)  &
<\infty,\nonumber\\
V(\rho\Vert\sigma)  &  >0. \label{eq:DTV-assumps}%
\end{align}
Then the following bound holds for all $\varepsilon\in(0,1)$ and sufficiently
large~$n$:%
\begin{multline}
D_{H}^{\varepsilon}(\rho^{\otimes n}\Vert\sigma^{\otimes n})\leq
\label{eq:hypo-expand}\\
nD(\rho\Vert\sigma)+\sqrt{nV(\rho\Vert\sigma)}\Phi^{-1}(\varepsilon)+O(\log
n).
\end{multline}

\end{proposition}

\begin{proof}
We follow the justification for Theorem~3 given in \cite{DPR15} closely, but
we do make some slight changes after the first few steps. Let $\Lambda^{n}$
be any measurement operator satisfying $\operatorname{Tr}\{\left(  I^{\otimes
n}-\Lambda^{n}\right)  \rho^{\otimes n}\}\leq\varepsilon$. By applying the
above proposition (making the replacements $\rho\rightarrow\rho^{\otimes n}$
and $\sigma\rightarrow\sigma^{\otimes n}$, so that $X_{n}$ is a sum of $n$
i.i.d.~random variables, each having mean $D(\rho\Vert\sigma)$, variance
$V(\rho\Vert\sigma)$, and third absolute central moment $T(\rho\Vert\sigma)$),
we find that%
\begin{align}
&  \operatorname{Tr}\{\Lambda^{n}\sigma^{\otimes n}\}\nonumber\\
&  \geq e^{-\theta_{n}}\left(  \frac{\Pr\{X_{n}\leq v_{n}\}}{1+e^{v_{n}%
-\theta_{n}}}-\operatorname{Tr}\{(I-\Lambda^{n})\rho^{\otimes n}\}\right) \\
&  \geq e^{-\theta_{n}}\left(  \frac{\Pr\{X_{n}\leq v_{n}\}}{1+e^{v_{n}%
-\theta_{n}}}-\varepsilon\right)  . \label{eq:start-point-CLT-bnd}%
\end{align}
The Berry--Esseen theorem \cite{KS10,S11} implies for any real number~$a$ that%
\begin{equation}
\Pr\left\{  \frac{ X_{n}-nD(\rho\Vert\sigma)}{\sqrt{nV(\rho\Vert\sigma)}}\leq
a\right\}  \geq\Phi(a) - K_{\rho,\sigma} \, n^{-1/2} ,
\end{equation}
where
\begin{equation}
K_{\rho,\sigma}\equiv\frac{C \, T(\rho\Vert\sigma)}{ [V(\rho\Vert
\sigma)]^{3/2}}%
\end{equation}
and $C \in(0,0.4748)$ \cite{KS10,S11}. It is clear that $K_{\rho,\sigma}$ is a
strictly positive constant $>C$ due to the assumption in
\eqref{eq:DTV-assumps} and the fact that $T(\rho\Vert\sigma) \geq[V(\rho
\Vert\sigma)]^{3/2}$ \cite{S11}. Let us set
\begin{multline}
v_{n}=nD(\rho\Vert\sigma)+\\
\sqrt{nV(\rho\Vert\sigma)}\Phi^{-1}(\varepsilon+ (2 + K_{\rho,\sigma})
n^{-1/2}),
\end{multline}
and note that we require sufficiently large $n$ here, so that the argument to
$\Phi^{-1}$ is $\in(0,1)$. We then find that%
\begin{equation}
\operatorname{Tr}\{\Lambda^{n}\sigma^{\otimes n}\}\geq e^{-\theta_{n}}\left(
\frac{\varepsilon+ 2 n^{-1/2}}{1+e^{v_{n}-\theta_{n}}}-\varepsilon\right)  .
\end{equation}
Now choosing $\theta_{n}=v_{n}+\frac{1}{2}\log n$, we get that%
\begin{multline}
\operatorname{Tr}\{\Lambda^{n}\sigma^{\otimes n}\}\geq\\
\left[  e^{-nD(\rho\Vert\sigma)-\sqrt{nV(\rho\Vert\sigma)}\Phi^{-1}%
(\varepsilon+ (2+ K_{\rho,\sigma}) n^{-1/2} )-\frac{1}{2}\log n}\right]
\times\\
\left(  \frac{1}{1+n^{-1/2}}\right)  ,
\end{multline}
so that the following inequality holds for sufficiently large~$n$:%
\begin{multline}
-\log\operatorname{Tr}\{\Lambda^{n}\sigma^{\otimes n}\}\leq\\
nD(\rho\Vert\sigma)+\sqrt{nV(\rho\Vert\sigma)}\Phi^{-1}(\varepsilon)+O(\log
n).
\end{multline}
In the last line, we have invoked \cite[Footnote~6]{TH12}, which in turn is an
invocation of Taylor's theorem: for $f$ continuously differentiable, $c$ a
positive constant, and $n \geq n_{0}$, the following equality holds
\begin{equation}
\sqrt{n} f(x + c/\sqrt{n}) = \sqrt{n}f(x) + c f^{\prime}(a)
\end{equation}
for some $a \in[x,x+c/\sqrt{n_{0}}]$.
\end{proof}

\section{Finiteness of the third absolute central moment of the log likelihood
ratio for quantum Gaussian states}

\label{sec:finiteness}

We argue in this final appendix that $T(\rho\Vert\sigma)$, the third absolute
central moment of the log-likelihood ratio of two finite-energy, zero-mean
Gaussian states $\rho$ and $\sigma$, is finite. By definition, we have that%
\begin{equation}
T(\rho\Vert\sigma)=\sum_{x,y}\left\vert \left\langle \psi_{y}|\phi
_{x}\right\rangle \right\vert ^{2}\lambda_{x}\left\vert \log_{2}\left(
\lambda_{x}/\mu_{y}\right)  -D(\rho\Vert\sigma)\right\vert ^{3},
\end{equation}
where the spectral decompositions of $\rho$ and $\sigma$ are given by%
\begin{equation}
\rho=\sum_{x}\lambda_{x}|\phi_{x}\rangle\langle\phi_{x}|,\qquad\sigma=\sum
_{y}\mu_{y}|\psi_{y}\rangle\langle\psi_{y}|.
\end{equation}
By concavity of $x^{3/4}$ for $x\geq0$, it follows that%
\begin{multline}
T(\rho\Vert\sigma)\\
=  \sum_{x,y}\left\vert \left\langle \psi_{y}|\phi_{x}\right\rangle
\right\vert ^{2}\lambda_{x}
\left[
\left\vert \log_{2}\left(  \lambda_{x}/\mu
_{y}\right)  -D(\rho\Vert\sigma)\right\vert ^{4}\right]  ^{3/4}\\
\leq\left[  \sum_{x,y}\left\vert \left\langle \psi_{y}|\phi_{x}\right\rangle
\right\vert ^{2}\lambda_{x}\left\vert \log_{2}\left(  \lambda_{x}/\mu
_{y}\right)  -D(\rho\Vert\sigma)\right\vert ^{4}\right]  ^{3/4}\\
=\left[  \sum_{x,y}\left\vert \left\langle \psi_{y}|\phi_{x}\right\rangle
\right\vert ^{2}\lambda_{x}\left(  \log_{2}\left(  \lambda_{x}/\mu_{y}\right)
-D(\rho\Vert\sigma)\right)  ^{4}\right]  ^{3/4},\label{eq:fourth-moment-bound}%
\end{multline}
and so we aim to show that this latter quantity is finite. For zero-mean,
$m$-mode faithful Gaussian states, the Williamson theorem \cite{W36} implies
that their spectral decompositions are as follows:%
\begin{align}
\rho &  =U_{\rho}\left(  \bigotimes\limits_{i=1}^{m}\theta(N_{\rho}%
^{i})\right)  U_{\rho}^{\dag},\\
\sigma &  =U_{\sigma}\left(  \bigotimes\limits_{i=1}^{m}\theta(N_{\sigma}%
^{i})\right)  U_{\sigma}^{\dag},
\end{align}
where $U_{\rho}$ and $U_{\sigma}$ denote Gaussian unitaries that can be
generated by a Hamiltonian no more than quadratic in the position- and
momentum-quadrature operators, $N_{\rho}^{i},N_{\sigma}^{i}>0$ for all $i$,
and $\theta(N)$ denotes a thermal state of mean photon number $N$:%
\begin{equation}
\theta(N)=\frac{1}{N+1}\sum_{n=0}^{\infty}\left(  \frac{N}{N+1}\right)
^{n}|n\rangle\langle n|,
\end{equation}
with $|n\rangle$ denoting a photonic number state. Introducing the multi-index
notation $|\vec{n}\rangle=|n_{1}\rangle\cdots|n_{m}\rangle$, we can then write
the overlap $\left\vert \left\langle \psi_{y}|\phi_{x}\right\rangle
\right\vert ^{2}$ as $\left\vert \langle\vec{l}|U_{\sigma}^{\dag}U_{\rho}%
|\vec{n}\rangle\right\vert ^{2}$. This conditional probability distribution
represents the probability of detecting the photon numbers $\vec{l}$ if the
photon number state $|\vec{n}\rangle$ is prepared and transmitted through the
Gaussian unitary $U_{\sigma}^{\dag}U_{\rho}\equiv V$. This distribution has
well defined (finite) higher moments with respect to photon number. Setting
$\hat{n}_{i}$ to be the photon number operator for the $i$th mode, this claim
follows because the $k$th moment of the conditional probability distribution
$\left\vert \langle\vec{l}|U_{\sigma}^{\dag}U_{\rho}|\vec{n}\rangle\right\vert
^{2}$ is given by%
\begin{multline}
\operatorname{Tr}\left\{  V|\vec{n}\rangle\langle\vec{n}|V^{\dag}\left(
\sum_{i=1}^{m}\hat{n}_{i}\right)  ^{k}\right\}  \label{eq:bounded-moments}\\
=\operatorname{Tr}\left\{  |\vec{n}\rangle\langle\vec{n}|\left(  \sum
_{i=1}^{m}V^{\dag}\hat{n}_{i}V\right)  ^{k}\right\}  .
\end{multline}
Since $V$ is a Gaussian unitary generated by a Hamiltonian no more than
quadratic in the position and momentum-quadrature operators \cite{S17}, each
$V^{\dag}\hat{n}_{i}V$ is a bounded linear combination of position and
momentum-quadrature operators and so $\left(  \sum_{i=1}^{m}V^{\dag}\hat
{n}_{i}V\right)  ^{k}$ is as well since $k$ is finite. Given that the photon
number states have bounded moments, we can conclude that
\eqref{eq:bounded-moments} is finite. The eigenvalues $\lambda_{x}$ and
$\mu_{y}$ in this case are given by%
\begin{align}
&  \prod\limits_{i=1}^{m}\left[  \frac{1}{N_{\rho}^{i}+1}\left(  \frac
{N_{\rho}^{i}}{N_{\rho}^{i}+1}\right)  ^{n_{i}}\right]  ,\label{eq:rho-dist}\\
&  \prod\limits_{i=1}^{m}\left[  \frac{1}{N_{\sigma}^{i}+1}\left(
\frac{N_{\sigma}^{i}}{N_{\sigma}^{i}+1}\right)  ^{l_{i}}\right]  ,
\end{align}
and indexed by the multi-indices $\vec{n}$ and $\vec{l}$, respectively. The
distribution in \eqref{eq:rho-dist}\ has well defined (finite)\ higher moments
with respect to photon number because it is a product of geometric
distributions. We can then write
$\log_{2}\!\left(  \lambda_{x}/\mu_{y}\right)
=\log_{2}\!\left(  \lambda_{x}\right)  -\log_{2}\!\left(  \mu_{y}\right)  $ as%
\begin{multline}
\sum_{i=1}^{m}\log_{2}\!\left(  \frac{N_{\sigma}^{i}+1}{N_{\rho}^{i}+1}\right)
+n_{i}\log_{2}\!\left(  \frac{N_{\rho}^{i}}{N_{\rho}^{i}+1}\right)  \\
-l_{i}\log_{2}\!\left(  \frac{N_{\sigma}^{i}}{N_{\sigma}^{i}+1}\right)  .
\end{multline}
Thus, after expanding, the last quantity in brackets in \eqref{eq:fourth-moment-bound} is
equal to an expression involving no more than the fourth moments of photon numbers, but we
have already argued that this is finite for the distributions under question.
As a consequence, we can conclude that $T(\rho\Vert\sigma)$ is finite whenever
$\rho$ and $\sigma$ are zero-mean, finite-energy, faithful Gaussian states.

\bibliographystyle{unsrt}
\bibliography{Ref}

\end{document}